\documentclass[a4paper]{article}

\usepackage[english]{babel}
\usepackage[ansinew]{inputenc}
\usepackage{amsmath}
\usepackage{amsthm}
\usepackage{amssymb}
\usepackage{amsfonts}
\usepackage{graphicx}

\usepackage[lofdepth,lotdepth]{subfig}
\usepackage{tikz}

\newcounter{saetning} 
  \newtheorem{theo}[saetning]{Theorem}
  \newtheorem{coro}[saetning]{Corollary}
  \newtheorem{lemm}[saetning]{Lemma}
  \newtheorem{prop}[saetning]{Proposition}
\theoremstyle{definition} 
	\newtheorem{defi}{Definition}
\theoremstyle{remark}
	
	\newtheorem{example}{Example}

\newcommand{\N}{\mathbb{N}}
\newcommand{\susp}{\text{susp}}

\newcommand{\E}{\mathbb{E}}

\newcommand{\plr}{\textsc{plr}}

\DeclareMathOperator{\Indep}{Indep}
\DeclareMathOperator{\Fixed}{Fixed}

\DeclareMathOperator{\Succ}{Succ}

\newcommand{\proto}[2]{#2^{#1}}

\title{Information Theoretical Cryptogenography}
\begin{document}
\date{}
\author{Sune K Jakobsen}
  \maketitle 
  
    \section*{Abstract}
    We consider problems where $n$ people are communicating and a random subset of them is trying to leak information, without making it clear who are leaking the information. We introduce a measure of suspicion, and show that the amount of leaked information will always be bounded by the expected increase in suspicion, and that this bound is tight. We ask the question: Suppose a large number of people have some information they want to leak, but they want to ensure that after the communication, an observer will assign probability $\leq c$ to the events that each of them is trying to leak the information. How much information can they reliably leak, per person who is leaking? We show that the answer is $\left(\frac{-\log(1-c)}{c}-\log(e)\right)$ bits. 
        
    \section{Introduction}
  The year is 2084 and the world is controlled by a supercomputer called Eve. It makes the laws, carries them out, has surveillance cameras everywhere, can hear everything you say, and can break any kind of cryptography. It was designed to make a world that maximises the total amount of happiness, while still being fair. However, Eve started to make some unfortunate decisions. For example, it thought that to maximise the utility it has been designed to maximise, it must ensure that it survives, so it decided to execute everyone it knew beyond reasonable doubt was trying to plot against Eve (it was designed so it could not punish anyone as long as there is reasonable doubt, and reasonable doubt had been defined to be a $5\%$ chance of being innocent). Everyone agrees that Eve should be shut down. The only person who can shut down Eve is Frank who is sitting in a special control room. Eve cannot hurt him, he has access to everything Eve can see, but he needs a password to shut down Eve. A small number of people, say $100$ Londoners, know the password.  Eve or Frank have no clue who they are, only that they exist. If one of them simply says the password, Eve will execute the person. So how can they reveal the password, without any of them getting killed?
  
  Suppose it is known that the password is the name of a museum in London. Frank then announces a date and time, and if you have the password, you show up at the correct museum that day and time, and if you do not have the information, you do as you would otherwise have done. If the museum is not too big, Frank will notice that there is one museum with more visitors than usual, so he gets the password. At the same time, if the museum is not too small, a large fraction of the visitors will just be there by chance, so Eve cannot punish any of them. 
  
  If the password is not necessarily the name of a museum, Frank can simply define a one-to-one correspondence between possible passwords and museums (or, if there are many possible passwords, take one letter at a time, with different people leaking each letter). We do not actually need museums to use this idea, the important part is that many people sends some messages, that will follow a fixed distribution if they do not think about it, and that if they want to, they can choose a specific message. For example, we could use parity of the minutes in the time we post messages on a blog. The purpose of this paper is to show how much information can be leaked this way.

   \subsection{Previous Work and Our Results}
   
  If we assume standard cryptographic assumptions, or if each pair of people had a private channel, we could use multi-party computation to let one person reveal information to a group of $n$ people, in such a way that if more than half of them follow the protocol, a computationally bounded observer will only have a negligible advantage when trying to guess which of the collaborating parties who originally had the information \cite{Goldreich87,Rabin89}. If we allow Frank to communicate, we could also use steganography \cite{Hopper04} to reveal the information to Frank, again only assuming standard cryptographic assumptions and that the observers have bounded computational power. 
  
  However, we assume that the observers have unbounded computational power, and that the observers see all messages sent. In that case, we could let every person sent random messages. People who knows the secret, $X$, could make their message correlated with $X$. For example the messages could be ``I think $X$ belongs to the set $S$''. However, every time you make a correct hint about what the secret $X$ is, it will increase the observers suspicion that you know $X$. The more precise the hint is or the more unlikely it is that you would give the hint without knowing $X$, the more useful the statement is to Frank. But at the same time, such statements would also be the statements that increases Eve suspicion towards you the most (at least if we assume she knows $X$). Our main contribution is to introduce a measure of suspicion that captures this, and to show that if you want to leak some amount of information about $X$ in the information theoretical sense, then your suspicion will, in expectation, have to increase by exactly the same amount. 
  
  The measure of suspicion turns out the be extremely useful for showing upper bounds on how much information you can leak, without making it clear that you are leaking. We show that if $n$ people are known to each know $X$ with probability $b$ independently of each other, and no one wants an observer to assign probability more than $c$ to the event that they were leaking information, they can each leak at most $\frac{-b\log(1-c)+c\log(1-b)}{c}$ bits about $X$. Using Shannon's Coding Theorem, we show that for all $\epsilon>0$ there exists $n$ such that if $X$ is uniformly distributed with entropy $\left(\frac{-b\log(1-c)+c\log(1-b)}{c}-\epsilon\right)n$ then $n$ such people can communicate in a way that would enable an observer to guess $X$ with probability $>1-\epsilon$, but for each person, the observer would still assign probability $\leq c$ to the event that that person was leaking. We show a similar result for the case where the total number of leakers is fixed and known.
  
  The measure of suspicion is also useful for analysing a generalisation of the original cryptogenography (hidden-origin-writing) problem, as introduced in \cite{Brody14}. Here the authors considered a game where one person among $n$ was randomly chosen and given the result of a coin flip. The goal for the $n$ players is to communicate in such a way that an observer, Frank, would guess the correct result of the coin flip, but another observer, Eve, who has the same information would guess wrong, when asked who of the $n$ originally knew the result of the coin flip. The main method in \cite{Brody14} is a concavity characterisation, and is very different from the information theory methods we use. We generalise the problem to $h$ bits of information and more players $l$ who have the information, and show that if $h=o(l)$ the winning probability tends to $1$ and if $l=o(h)$ the winning probability tends to $0$. 

   Finally we show that in general to do cryptogenography, you do not need the non-leakers to collaborate. Instead we can use the fact that people send out random messages anyway, and use this in a similar way to steganography (see \cite{Hopper04}). All we need is that people are communicating in a way that involve sufficiently randomness and that they do not change this communication, when we build a protocol on top of that. We can for example assume that they are not aware of the protocol, or they do not care about the leakage.

  \subsection{Paper Outline}
  We define notation and recall some concepts and theorems from information theory and introduce a communication model in Section \ref{sec:prelim}. In Section \ref{sec:mutual} we introduce a measure of suspicion and use this to show upper bounds on how much information the players can leak if they want Eve to have reasonable doubt that they are leaking. In Section \ref{sec:reliable} we turn to reliable leakage, and define and determine the capacity for some cryptogenography problems. In Section \ref{sec:org} we show how our results can be used to analyse a generalisation of the original cryptogenography problem.  Finally, in Section \ref{sec:innocents} we show that we can do equally well, even if the non-leakers are not collaborating in leaking, but are just communicating innocently.

  \section{Preliminaries} \label{sec:prelim}
  Unless stated otherwise, all random variables in this paper are assumed to be discrete. Random variables are denoted by capital letters and their support are denoted by the calligraphic version of the same letter (e.g. $\mathcal{X}$ is the support of $X$). 
 If $X$ and $Y$ are random variables and $\Pr(Y=y)>0$, we let $X|_{Y=y}$ denote the random variable $X$ conditioned on $Y=y$. That is
 \[\Pr(X|_{Y=y}=x)=\frac{\Pr(X=x,Y=y)}{\Pr(Y=y)}.\]
 For a tuple or infinite sequence $a$, we let $a_i$ denote the $i$'th element of $a$, and let $a^i=(a_1,\dots, a_i)$ be the tuple of the $i$ first elements from $a$. Similarly if $A$ is a tuple or sequence of random variables. For a tuple $a$ of $n$ elements we let $a\circ a'$ denote the tuple $(a_1,\dots, a_n,a')$. 
 
  For a random variable $X$ and a value $x\in\mathcal{X}$ with $\Pr(X=x)>0$ the \emph{surprisal} or the \emph{code-length}\footnote{If $-\log(\Pr(X=x))$ is an integer for all $x\in\mathcal{X}$, and we want to find an optimal prefix-free binary code for $X$, the length of the code for $x$ should be $-\log(\Pr(X=x))$, thus the name code-length. If they are not integers, we can instead use $\lceil -\log(\Pr(X=x))\rceil$ and waste at most one bit.} of $x$ is given by
  \begin{align}-\log(\Pr(X=x)),\nonumber
  \end{align}
  where $\log$, as in the rest of this paper, is the base-$2$ logarithm.
  
  The \emph{entropy of $X$}, $H(X)$, is the expected code-length of $X$
  \begin{align}H(X)=&\E -\log(\Pr(X=x))\nonumber\\
  =&-\sum_{x\in\mathcal{X}}\Pr(X=x)\log(\Pr(X=x)),\nonumber
  \end{align}
 where we define $0\log(0)=0$. If $p,q:\mathcal{X}\to[0,1]$ are two probability distributions on $\mathcal{X}$ we have the inequality
  \begin{align}
 -\sum_{x\in\mathcal{X}}p(x)\log(p(x))\leq -\sum_{x\in\mathcal{X}}p(x)\log(q(x)),\label{ineq:optimalcode}
 \end{align}
 with equality if and only if $p=q$  \cite{ThomasCover91}. The interpretation is, if $X$'s distribution is given by $p$, and you encode values of $X$ using a code optimised to the distribution $q$, you get the shortest average code-length if and only if $p=q$. 
 
 The entropy of a random variable $X$ can be thought of as the uncertainty about $X$, or as the amount of information in $X$. For a tuple of random variables $(X_1,\dots X_k)$ the entropy $H(X_1,\dots X_k)$ is simply the entropy of the random variable $(X_1,\dots,X_k)$. The \emph{entropy of $X$ given $Y$}, $H(X|Y)$ is
  \begin{align}
  H(X|Y)=&\sum_{y\in\mathcal{Y}}\Pr(Y=y)H(X|_{Y=y}).\label{eq:givenH}
  \end{align}
  A simple computation shows that
  \begin{align}
  H(X|Y)=&H(X,Y)-H(Y).\nonumber
  \end{align}
  The \emph{mutual information} $I(X;Y)$ of two random variables $X,Y$ is given by
  \begin{align}I(X;Y)=H(X)+H(Y)-H(X,Y)=H(Y)-H(Y|X).\nonumber\end{align}
  This is known to be non-negative. The \emph{mutual information $I(X;Y|Z=z)$ of $X$ and $Y$ given $Z=z$} is given by
  \begin{align}
  I(X;Y|Z=z)=I(X|_{Z=z};Y|_{Z=z}),\nonumber
  \end{align}
  where the joint distribution of $(X|_{Z=z},Y|_{Z=z})$ is given by $(X,Y)|_{Z=z}$. The \emph{mutual information $I(X;Y|Z)$ of $X$ and $Y$ given $Z$} is 
  \begin{align}
  I(X;Y|Z=z)=\E_zI(X;Y|Z=z).\nonumber
  \end{align} 
  A simple computation shows that
  \begin{align}
  I(X;Y|Z)=H(X,Z)+H(Y,Z)-H(X,Y,Z)-H(Z).\nonumber
  \end{align}
  We will need the chain rule for mutual information,
  \begin{align}
  I(X;(T_1,\dots T_k))=\sum_{i=1}^k I(X;T_i|(T_1,\dots, T_{i-1})).\nonumber
  \end{align}
  Let $X$ and $Y$ be random variables, and $f:\mathcal{Y}\to\mathcal{X}$ a function. We think of $f(Y)$ as a guess about what $X$ is. The probability of error, $P_e$ is now $\Pr(f(Y)\neq X)$. We will need (a weak version of) Fano's inequality,
  \begin{align}
  P_e\geq \frac{H(X|Y)-1}{\log(|\mathcal{X}|)}.\label{ineq:Fano}
  \end{align}
  A \emph{discrete memoryless channel} (or \emph{channel} for short) $q$ consist of a finite \emph{input set} $\mathcal{Y}$, a finite \emph{output set} $\mathcal{Z}$ and for each element $y\in \mathcal{Y}$ of the input set a probability distribution $q(z|y)$ on the output set. If Alice have some information $X$ that she wants Bob to know, she can use a channel. To do that, Alice and Bob will have to both know a code. An \emph{error correcting code}, or simply a \emph{code}, $\mathfrak{C}:\mathcal{X}\to\mathcal{Y}^n$ is a function that for each $x\in\mathcal{X}$ specifies what Alice should give as input to the channel. Here $n$ is the \emph{length} of the code. Now the probability that Bob receive $Z_{\mathfrak{C}}=z_1\dots z_n$ when $X=x$ is given by
  \[\Pr(Z_{\mathfrak{C}}=z|X=x)=\prod_{i=1}^nq\left(z_i|\mathfrak{C}(x)_i\right).\]
  When Bob knows $q$, $\mathfrak{C}$ and the distribution of $X$ he can compute $\Pr(X=x|Z_{\mathfrak{C}}=z)$. Let $\hat{X}$ denote the most likely value of $X$ given $Z_{\mathfrak{C}}$. A \emph{rate} $R$ is \emph{achievable} if for all $\epsilon>0$ there is a $n>0$ such that for $X$ uniformly distributed on $\{1,\dots, 2^{\lceil Rn\rceil}\}$ there is a code $\mathfrak{C}$ of length $n$ for $q$ giving $\Pr(\hat{X}=x|X=x)>1-\epsilon$ for all $x\in\mathcal{X}$. 
  
  For a distribution $p$ on the input set $\mathcal{Y}$ we get a joint distribution of $(Y,Z)$ given by $\Pr(Y=y,Z=z)=p(y)q(z|y)$. Now define the capacity $C$ of $q$ to be 
  \[C=\max_{p}I(Y;Z),\]
  where $\max$ is over all distributions $p$ of $Y$ and the joint distribution of $(Y,Z)$ is as above. Shannon's Noisy Coding Theorem says that any rate below $C$ is achievable, and no rate above $C$ is achievable \cite{Shannon48}. For an introduction to these information theoretical concepts and for proofs, see \cite{ThomasCover91}. 
  
  \subsection{Model}\label{subsec:model}
  
  In this paper we consider problems where one or more players might be trying to leak information about the outcome of a random variable $X$. The number of players is denoted $n$ and the players are called $\plr_1,\dots, \plr_n$. Sometimes we will call $\plr_1$ Alice and $\plr_2$ Bob.  We let $L_i$ be the random variable that is $1$ if player $i$ knows the information and $0$ otherwise. If there is only one player we write $L$ instead of $L_1$. The joint distribution of $(X,L_1,\dots, L_n)$ is known to everyone. 
  
  All messages are broadcasted to all players and to two observers, Eve and Frank. The two observers will have exactly the same information, but we will think of them as two people rather than one. We want to reveal information about $X$ to Frank, while at the same time make sure that for all $i$, Eve does not get too sure that $L_i=1$. The random variable that is the transcript of a protocol will be denoted $T$, and specific transcripts $t$. This is a tuple of messages, so we can use the notation $T^k, T_k,t^k,t_k$ as define in the beginning of this section. For example, $T^k$ denotes the tuple of the first $k$ messages.
   
 In this section we define the collaborating model. In Section \ref{sec:innocents} we will define a model, were we do not need the non-leakers to collaborate. The model in Section \ref{sec:innocents} will be more useful in practice, however when constructing protocols, it is easier first to construct them in the collaboration model. In the \emph{collaborating model} we can tell all the players including the non-leaking players to follow some communication protocol, called a collaborating cryptogenography protocol. The messages send by leaking player may depend on the value of $X$, but the messages of non-leaking players have to be independent of $X$ given the previous transcript. Formally, a \emph{collaborating cryptogenography protocol} $\pi$ specifies for any possible value $t^k$ of the current transcript $T^k$:
     \begin{itemize}
\item Should the communication stop or continue, and if it should continue,
\item Who is next to send a message, say $\plr_i$, and 
\item A distribution $p_?$ and a set of distributions, $\{p_x\}_{x\in \mathcal{X}}$ (the distributions $p_?$ and $\{p_x\}_{x\in \mathcal{X}}$ depend on $\pi$ and $t^k$). Now $\plr_i$ should choose a message using $p_?$, if $L_i=0$ and choose a message using $p_x$ if $L=1$ and $X=x$.
\end{itemize}
Furthermore, for any protocol $\pi$, there should a number $length(\pi)$ such that the protocol will always terminate after at most $length(\pi)$ messages. We assume that both Frank and Eve know the protocol. They also know the prior distribution of $(X,L_1,\dots,L_n)$, and we assume that they have  computational power to compute $(X,L_1,\dots,L_n)|_{T=t}$ for any transcript $t$. Notice that this assumption rules out the use of cryptography. 

One way that everyone can know the protocol, is if one person, e.g. Frank, announces the protocol that they will use, and we assume that everyone follows that protocol. Another possibility is that the players and Frank and Eve (or their ancestors) have played a game about leaking information many times and slowly developed (or evolved) a protocol for leaking information and learned (or evolved) to play the game optimally. In this paper we will not consider the question of if and how the protocol could be developed or evolved.

While we think of different players as different people, two or more different players could be controlled by the same person. For example, if they are communicating in a chatroom with perfect anonymity, except that a profile's identity will be revealed if the profile can be shown to be guilty in leaking with probability $>95\%$. Here each player would correspond to a profile, but the same person could have more profiles. However, we will use ``player'' and ``person'' as synonyms in the paper.

  \section{Bounds on $I(X;T)$}\label{sec:mutual}
  
  \subsection{Suspicion}

 First we will look at the problem where only one player is communicating and she may or may not be trying to leak information. We will later use these results when we analyse the many-player problem.
    
    In the one player case, Alice sends one message $A$. If she is not trying to leak information, she will choose this message in $\mathcal{A}$ randomly using a distribution $p_?$. If she is trying to leak information, and $X=x$, she will use a distribution $p_x$. For a random variable $Y$ and a value $y\in\mathcal{Y}$ with $\Pr(Y=y)>0$ we let $c_{Y=y}=\Pr(L=1|Y=y)$. We usually suppress the random variable, and write $c_y$ instead. Here $Y$ could be a tuple of random variables, and $y$ a tuple of values. If $y=(y_1,y_2)$ is a tuple, we write $c_{y_1y_2}$ instead of $c_{(y_1,y_2)}$.
    
    We want to see how much information Alice can leak to Frank (by choosing the $p$'s), without being too suspicious to Eve. The following measure of suspicion turns out to be useful. 
    
    \begin{defi} Let $Y$ be a random variable jointly distributed with $L$. Then the \emph{suspicion (of Alice) given $Y=y$} is 
    \begin{align}\susp(Y=y)=&-\log(1-c_y)\nonumber\\
    =&-\log(\Pr(L=0|Y=y)).\nonumber
    \end{align}
    \end{defi}
    We see that $\susp(Y=y)$ depends on $y$ and the joint distribution of $L$ and $Y$, but to keep notation simple, we suppress the dependence on $L$. The suspicion of Alice measures how suspicious Alice is to someone who knows that $Y=y$ and knows nothing more. For example $Y$ could be the tuple that consists of the secret information $X$ and the current transcript.
    
    We can think of the suspicion as the surprisal of the event, ``Alice did not have the information''. Next we define the suspicion given a random variable $Y$, without setting it equal to something. 
    \begin{defi}
    The \emph{suspicion (of Alice) given $Y$} is 
    \begin{align}
    \susp(Y)=&\E_y \susp(Y=y)\nonumber\\
    =&\sum_{y\in\mathcal{Y}} \Pr(Y=y)\susp(Y=y).\label{eq:suspisaveg}
    \end{align}
    In each of these, $Y$ can consist of more than one random variable, e.g. $Y=(X,A)$. Finally we can also combine these two definitions, giving
    \begin{align}\susp(X,A=a)=\sum_{x\in \mathcal{X}}\Pr(X=x|A=a)\susp((X,A)=(x,a)).\nonumber\end{align}
    Where $X$ and $A$ can themselves be tuples of random variables.
    \end{defi}
    
    The definitions imply that
    \[\susp(X,A)=\sum_{a\in\mathcal{A}}\Pr(A=a)\susp(X,A=a),\]
    which can be thought of as (\ref{eq:suspisaveg}) given $X$. 
    
    When Alice sends a message $A$ this might reveal some information about $X$, but at the same time, she will also reveal some information about whether she is trying to leak $X$. We would like to bound $I(A;X)$ by the information $A$ reveals about $L$. This is not possible. If, for example, we set $A=X$ whenever $L=1$ and $A=a\not  \in \mathcal{X}$ when $L=0$, then $I(A;X)=\Pr(L=1)H(A)$ which can be large, but $I(A;L)\leq H(L)\leq 1$. The theorem below shows that instead, $I(A;X)$ can be bounded by the expected increase in suspicion given $X$, and that this bound is tight. 
    
         \begin{theo}\label{theo:Alice}
     If Alice sends a message $A$, we have 
     \begin{align}
     I(X;A)\leq \susp(X,A)-\susp(X).\nonumber
     \end{align}
     That is, the amount of information she sends about $X$ is at most her expected increase in suspicion given $X$. There is equality if and only if the distribution of $A$ is the same as $A|_{L=0}$. 
     \end{theo}
    
    \begin{proof}
    With no information revealed, Alice's suspicion given $X$ is
    \begin{align}
    \susp(X)=-\sum_{x\in\mathcal{X}}\Pr(X=x)\log(1-c_x).\nonumber
    \end{align}
    We want to compute Alice's suspicion given $X$ and her message $A$. 
    \begin{align}
    \susp(X,A)=&\sum_{x,j}\Pr(X=x,A=a)\susp(X=x,A=a)\nonumber\\
    =&-\sum_{x,a}\Pr(X=x,A=a)\log(1-c_{xa})\nonumber\\
    =&-\sum_{x,a}\Pr(X=x,A=a)\left(\log(1-c_x)+\log\left(\frac{1-c_{xa}}{1-c_x}\right)\right).\nonumber
    \end{align}
    
     Now it follows that the cost in suspicion given $X$ of sending $A$ is
     \begin{align}
     \susp(X,A)-\susp(X)=-\sum_{x,j}\Pr(X=x,A=a)\log\left(\frac{1-c_{xa}}{1-c_x}\right).\label{eq:suspincr}
     \end{align}
 
     Next we want to see how much information $A$ gives about $X$, that is $I(A;X)=H(A)-H(A|X)$. We claim that this is bounded by the cost in suspicion, or equivalently, $H(A)\leq \susp(X,A)-\susp(X)+H(A|X)$. First we compute $H(A|X)$ using (\ref{eq:givenH}):
      \begin{align}
     H(A|X)=&\sum_{x}\Pr(X=x)H(A|X=x)\nonumber\\
     =& -\sum_{x}\Pr(X=x)\sum_a \Pr(A=a|X=x)\log(\Pr(A=a|X=x))\nonumber\\
     =& -\sum_{x,a}\Pr(X=x,A=a)\log(\Pr(A=a|X=x)).\label{eq:AgivenX}
     \end{align}

     We have
     \begin{align}
     \frac{1-c_{xa}}{1-c_x}&\Pr(A=a|X=x)\nonumber\\
     =&\frac{\Pr(L=0|X=x,A=a)}{\Pr(L=0|X=x)}\Pr(A=a|X=x)\nonumber\\
     =&\frac{\Pr(L=0,X=x,A=a)}{\Pr(X=x,A=a)}\frac{\Pr(X=x)}{\Pr(L=0,X=x)}\frac{\Pr(X=x,A=a)}{\Pr(X=x)}\nonumber\\
     =&\frac{\Pr(L=0,X=x,A=a)}{\Pr(L=0,X=x)}\nonumber\\
     =&\Pr(A=a|X=x,L=0)\nonumber\\
     =&\Pr(A=a|L=0)\label{eq:cxaca}
     \end{align}
     Here, the last equation follows from the assumption that $A$ does not depend on $X$ when $L=0$.     
     From this we conclude
    \begin{align*}
    \susp(X,A)&-\susp(X)+H(A|X)\nonumber\\
    =&-\sum_{x,a}\Pr(X=x,A=a)\log\left(\frac{1-c_{xa}}{1-c_x}\Pr(A=a|X=x)\right)\nonumber\\
    =&-\sum_{x,a}\Pr(X=x,A=a)\log\left(\Pr(A=a|L=0)\right)\nonumber\\
     =&-\sum_{a}\Pr(A=a)\log\left(\Pr(A=a|L=0)\right)\nonumber\\
     \geq& -\sum_{a}\Pr(A=a)\log(\Pr(A=a))\nonumber\\
     =&H(A).
     \end{align*}   
     Here the first equality follows from (\ref{eq:suspincr}) and (\ref{eq:AgivenX}), the second follows from (\ref{eq:cxaca}) and the inequality follows from inequality (\ref{ineq:optimalcode}). There is equality if and only if $\Pr(A=a)=\Pr(A=a|L=0)$ for all $a$.
      \end{proof}
   
     We will now turn to the problem where many people are communicating. We assume that they sent messages one at a time, so we can break the protocol into time periods were only one person is communicating, and see the entire protocol as a sequence of one player protocols. The following Corollary show that a statement similar to Theorem \ref{theo:Alice} holds for each single message in a protocol with many players. 
   
   \begin{coro}\label{coro:Alice}
   Let $(L,T^{k-1},X)$ have some joint distribution, where $T^{k-1}$ denotes previous transcript. Let $T_k$ be the next message sent by Alice. Then 
   \[I(X;T_k|T^{k-1})\leq \susp(X,T^k)-\susp(X,T^{k-1}).\]
   \end{coro}
   \begin{proof}
   For a particular value $t^{k-1}$ of $T^{k-1}$ we use Theorem \ref{theo:Alice} with $(X,T_k)|_{T^{k-1}=t^{k-1}}$ as $(X,A)$ to get
   \[I(X;T_k|T^{k-1}=t^{k-1})\leq \susp(X,T_k,T^{k-1}=t^{k-1})-\susp(X,T^{k-1}=t^{k-1}).\]
   By multiplying each side by $\Pr(T^{k-1}=t^{k-1})$ and summing over all possible $T^{k-1}$ we get the desired inequality.
   \end{proof}
 
     
   A protocol consists of a sequence of messages that each leaks some information and increases the suspicion of the sender. We can add up increases in suspicion, and using the chain rule for mutual information we can also add up the amount of revealed information. However, we have to be aware that Bob's message not only affect his own suspicion, but it might also affect Alice's suspicion. To show an upper bound on the amount of information a group of people can leak, we need to show that one persons message will, in expectation, never make another persons suspicion decrease. We get this from the following proposition by setting $Y=(X,T^{k-1})$ and $B=T_k$.

     \begin{prop}\label{prop:Bob}
     For any joint distribution on $(L,Y,B)$ we have $\susp(Y)\leq \susp(Y,B)$.
     \end{prop}
     \begin{proof}
     We have
      \begin{align}
     \susp(Y=y)=&-\log(\Pr(L=0|Y=y))\nonumber\\
     =& -\log\left(\sum_{b\in\mathcal{B}}\Pr(B=b|Y=y)\Pr(L=0|Y=y,B=b)\right)\nonumber\\
     \susp(Y=y,B)=&-\sum_{b\in\mathcal{B}}\Pr(B=b|Y=y)\log\Pr(L=0|Y=y,B=b).\nonumber\\
     \end{align}
     As $p\mapsto -\log(p)$ is convex, Jensen's inequality gives us
     \begin{align}
     \susp(Y=y,B)\geq \susp(Y=y).\nonumber
     \end{align}
    Multiplying each side by $\Pr(Y=y)$ and summing over all $y\in \mathcal{Y}$ gives us the desired inequality. 
    \end{proof}

In the proof of the next theorem we will assume that the protocol runs for a fixed number of messages, and the player to talk in round $k$ only depends on $k$, not on which previous messages was send. Any protocol $\pi$ can be turned into such a protocol $\pi'$ by adding dummy messages: In round $k$ of $\pi'$ we let $\plr_{k\text{ mod }n}$ talk. They follow protocol $\pi$ in the sense that if it is not $\plr_{k\text{ mod }n}$ turn to talk according to $\pi$ she send some fixed message $1$, and if it is her turn, she chooses her message as in $\pi$.

Let $\susp_i$ denote the suspicion of $\plr_i$.\footnote{This is defined similar to the suspicion of Alice, except using $L_i$ instead of $L$.}
\begin{theo}\label{theo:leakleqsusp}
If $T$ is the transcript of the entire protocol we have
\[I(X;T)\leq \sum_{i=1}^n \left(\susp_i(X,T)-\susp_i(X) \right).\]
\end{theo}
\begin{proof}
From the chain rule for mutual information, we know that
\[I(X;T)=\sum_{k=1}^{length(\pi)} I(X;T_k|T^{k-1}).\] 
Now Corollary \ref{coro:Alice} shows that $I(X;T_k|T^{k-1})\leq \susp_i(X,T^k)-\susp_i(X,T^{k-1})$ if $\plr_i$ send the $k$th message and Proposition \ref{prop:Bob} shows that $ \susp_{i'}(X,T^k)\geq \susp_{i'}(X,T^{k-1})$ for all other $i'$. Summing over all rounds in the protocol, we get the theorem.
\end{proof}

\subsection{Keeping reasonable doubt}

Until now we have bounded the amount of information the players can leak by the expected increase in some strange measure, suspicion, that we defined for the purpose. But there is no reason to think that someone who is leaking information cares about the expected suspicion towards her afterwards. A more likely scenario, is that each person leaking wants to ensure that after the leakage, an observer will assign probability at most $c$ to the event that she was leaking information. If this is the case after all possible transcripts $t$, we see that $\susp_i(X,T)\leq -\log(1-c)$. If we assume that each player before the protocol had probability $b<c$ of leaking independently of $X$, that is $\Pr(L_i|X=x)=b$ for all $x$ and $i$, we have $\susp_i(X)=-\log(1-b)$. Thus
\begin{align}
I(X;T)\leq \sum_{i=1}^n\left(\susp_i(X,T)-\susp_i(X)\right)=\left(\log(1-c)+\log(1-b)\right)n.\label{ineq:badUpper}
\end{align}
To reach this bound, we would need to have $\Pr(L_i=1|X=x,T=t)=c$ for all $x,t,i$. But the probability $\Pr(L_i=1|X=x)=b$ can also be computed as $\E_t \Pr(L_i=1|X=x,T=t)$, so $\Pr(L_i=1|X=x,T=t)$ cannot be constantly $c>b$. The following theorem improves the upper bound from  (\ref{ineq:badUpper}) by taking this into account. 

\begin{theo}\label{theo:generalupperbound}
Let $\pi$ be a collaborating cryptogenography protocol, and $T$ be its transcript. If for all players $\plr_i$ and all $x\in\mathcal{X}$ and all transcripts $t$ we have $\Pr(L_i=1|X=x)=b$, and $\Pr(L_i=1|T=t,X=x)\leq c$ then
\[I(X;T)\leq \frac{-b\log(1-c)+c\log(1-b)}{c}n.\]
\end{theo}

\begin{proof}
If $\Pr(L_i=1|X=x,T=t)\leq c$ then 
\begin{align}
\susp_i(X=x,T=t)=&-\log(1-\Pr(L_i=1|X=x,T=t))\nonumber\\
\leq &\frac{-\log(1-c)}{c}\Pr(L_i=1|X=x,T=t). \label{ineq:boundsusp}
\end{align}
This follows from the fact that we have equality when $\Pr(L_i=1|X=x,T=t)$ is $0$ or $c$, and the left hand side is convex in $\Pr(L_i=1|X=x,T=t)$ while the right hand side is linear.

 Let $\pi$ and $T$ be as in the assumptions. Now we get
\begin{align*}
\susp_i(X,T)=&\sum_{x,t}\Pr(X=x,T=t)\susp_i(X=x,T=t)\\
\leq& \sum_{x,t}\Pr(X=x,T=t)\frac{-\log(1-c)}{c}\Pr(L_i=1|X=x,T=t)\\
=&\sum_{x,t}\frac{-\log(1-c)}{c}\Pr(L_i=1,X=x,T=t)\\
=&\frac{-\log(1-c)}{c}\Pr(L_i=1)\\
=&\frac{-b\log(1-c)}{c}.
\end{align*}

Thus 
\begin{align*}
I(X;T)\leq &\sum_{i=1}^n \left(\susp_i(X,T)-\susp_i(X) \right)\\
\leq &\left(\frac{-b\log(1-c)}{c}-(-\log(1-b))\right)n\\
=&\frac{-b\log(1-c)+c\log(1-b)}{c}n.
\end{align*}
\end{proof}


It is clear that the upper bound from Theorem \ref{theo:generalupperbound} cannot be achieved for all distributions of $(X,L_1,\dots, L_n)$. If for example $H(X)<\frac{-b\log(1-c)+c\log(1-b)}{c}n$ we must also have $I(X,T)\leq H(X)<\frac{-b\log(1-c)+c\log(1-b)}{c}n$, that is, the players do not have enough information to send to reach the upper bound. Even if $H(X)$ is high, we may not be able to reach the upper bound. If it is known that $L_1=L_2=\dots=L_n$ the suspicion of the players will not depend on the player, only on the messages sent. So this problem will be equivalent to the case where only one person is sending messages.

We will now give an example where the upper bound from Theorem \ref{theo:generalupperbound} is achievable. We will refer back to this example when we prove that reliable leakage is possible. 

\begin{example}\label{exam:achieve}

Assume that $X,L_1,\dots, L_n$ are all independent, and $\Pr(L_i=1)=b$ for all $i$. Furthermore, assume that $0<b<c<1$ and that $\frac{b(1-c)}{c(1-b)}$ is a rational number. Let $d,a\in \N$ be the smallest natural numbers such that $\frac{a}{d}=\frac{b(1-c)}{c(1-b)}$. We see that $\frac{b(1-c)}{c(1-b)}\in (0,1)$ so $0<a<d$. We will assume that $X$ is uniformly distributed on $\{1,\dots, d\}^n$. 

Each player $\plr_i$ now sends one message, independently of which messages the other players send. If $L_i=0$, $\plr_i$ chooses a message in $\{1,\dots ,d\}$ uniformly at random. If $L_i=1$ and $X_i=x_i$, then $\plr_i$ chooses a message in 
\[\{1+(x_i-1)a,2+(x_i-1)a\dots,x_ia\} \mod{d}\] 
uniformly at random.\footnote{We use $k\mod{d}$ to mean the number in $\{1,\dots d\}$ that is equal to $k$ modulo $d$.}

We see that over random choice of $X$, the message, $A_i$, that $\plr_i$ sends, is uniformly distributed on $\{1,\dots, d\}$, so $H(A_i)=\log(d)$. We want to compute $H(A_i|X)$. Given $X$, each of the $d-a$ elements not in $\{1+(x_i-1)a,2+(x_i-1)a\dots,x_ia\} \mod{d}$ can only be send if $L=0$, so they will be send with probability $\frac{1-b}{d}$. Each of the $a$ elements in the set $\{1+(x_i-1)a,2+(x_i-1)a\dots,x_ia\} \mod{d}$ are sent with probability $\frac{b}{a}+\frac{1-b}{d}$. Thus
\begin{align*}
H(A_i|X) =& -\sum_{t_i\in \mathcal{A}_i} \Pr(A_i=t_i)\log(\Pr(A_i=t_i))\\
=&-a\left(\frac{b}{a}+\frac{1-b}{d}\right)\log\left(\frac{b}{a}+\frac{1-b}{d}\right)-(d-a)\frac{1-b}{d}\log\left(\frac{1-b}{d}\right)\\
=& -\frac{b}{c}\log\left(\frac{1-b}{d(1-c)}\right)-\left(1-\frac{b}{c}\right)\log\left(\frac{1-b}{d}\right).
\end{align*}
The last equality follows from three uses of $\frac{a}{d}=\frac{b(1-c)}{c(1-b)}$, or of its equivalent formulation, $\frac{b}{a}+\frac{1-b}{d}=\frac{b}{ac}$. Now
\begin{align}
I(A_i;X)=&H(A_i)-H(A_i|X)\nonumber\\
=&\log(d)+\frac{b}{c}\log\left(\frac{1-b}{d(1-c)}\right)+\left(1-\frac{b}{c}\right)\log\left(\frac{1-b}{d}\right)\nonumber\\
=&\log(1-b)-\frac{b}{c}\log(1-c)\nonumber\\
=&\frac{-b\log(1-c)+c\log(1-b)}{c}.\label{eq:capa}
\end{align}
The tuples $(X_i,A_i,L_i)$ where $i$ ranges over $\{1,\dots n\}$ are independent from each other, so we have $I(T;X)=\frac{-b\log(1-c)+c\log(1-b)}{c}n$ as wanted. 

Next we want to compute $\Pr(L_i=1|T=t,X=x)$. This is $0$ if $\plr_i$ send a message not in $\{1+(x_i-1)a,2+(x_i-1)a\dots,x_ia\} \mod{d}$. Otherwise we use independence and then Bayes' Theorem to get
\begin{align}
\Pr(L_i=1|T=t,X=x)=&\Pr(L_i=1|A_i=t_i,X_i=x_i)\nonumber\\
=&\frac{\Pr(A_i=t_i|L_i=1,X_i=x_i)\Pr(L_i=1|X_i=x_i)}{\Pr(A_i=t_i|X_i=x_i)}\nonumber\\
=&\frac{\frac{1}{a}b}{\frac{b}{a}+\frac{1-b}{d}}\nonumber\\
=&\frac{\frac{b}{a}}{\frac{b}{ac}}\nonumber\\
=&c.\label{eq:Bayes}
\end{align}
As we wanted. 
  \end{example}

  \section{Reliable leakage}\label{sec:reliable}
  
  In the previous example, Frank would receive some information about $X$ in the sense of information theory: Before he sees the transcript, any value of $X$ would be as likely as any other value, and when he knows the transcript, he has a much better idea about what $X$ is. However, his best guess about what $X$ is, is still very unlikely to be correct. Next we want to show that we can have reliable leakage. That is, no matter what value $X$ is taking, we want Frank to be able to guess the correct value with high probability. We will see that this is possible, even when $X$ have entropy close to $\frac{-b\log(1-c)+c\log(1-b)}{c}n$. Frank's guess would have to be a function $D$ of the transcript $t$. Saying that Frank will guess $X$ correct with high probability when $X=x$ is that same as saying that  $\Pr(D(T)=x|X=x)$ is close to one. 
    
  \begin{defi}
  Let $L=(L_1,\dots, L_n)$ be a tuple of random variables, where the $L_i$ takes values in $\{0,1\}$.
  
  A \emph{risky $(n,h,L,c,\epsilon)$-protocol} is a collaborating cryptogenography protocol together with a function $D$ from the set of possible transcripts to $\mathcal{X}=\{1,\dots,2^{\lceil h\rceil}\}$ such that when $X$ and $L$ are distributed independently and $X$ is uniformly distributed on $\mathcal{X}$, then for any $x\in\mathcal{X}$, there is probability $1-\epsilon$ that a random transcript $t$ distributed as $T|_{X=x}$ satisfies
    \begin{itemize}
	\item $\forall i: \Pr(L_i=1|T=t,X=x)\leq c$, and
	\item $D(t)=x$
    \end{itemize}
  \end{defi}
  That is, no matter the value of $X$, with high probability Frank can guess the value of $X$, and with high probability no player will be estimated to have leaked the information with probability $>c$ by Eve. However, there might be a small risk that someone will be estimated to have leaked the information with probability $>c$. This is the reason we call it a risky protocol. A safe protocol is a protocol where this never happens. 
  \begin{defi}
 A \emph{safe $(n,h,L,c,\epsilon)$-protocol} is a risky $(n,h,L,c,\epsilon)$-protocol where $\Pr(L_i=1|T=t,X=x)\leq c$ for all $i,t,x$ with $\Pr(T=t,X=x)>0$. 
  \end{defi}
  
  First we will consider the case where $L_1,\dots, L_n$ are independent, and the $L_i$'s all have the same distribution.
  
  \begin{defi}\label{defi:Indep}
  Let $\Indep_{b}(n)$ be the random variable $(L_1,\dots, L_n)$ where $L_1,\dots,L_n$ are independent, and each $L_i$ is distributed on $\{0,1\}$ and $\Pr(L_1=1)=b$.
  
  A rate $R$ is \emph{safely/riskily  $c$-achievable for $\Indep_b$} if for all $\epsilon>0$ and all $n_0$, there exists a safe/risky $(n,nR,\Indep_b(n),c,\epsilon)$-protocol with $n\geq n_0$.
  
  The \emph{safe/risky $c$-capacity for $\Indep_b$} is the supremum of all safely/riskily $c$-achievable rates for $\Indep_b$.
  \end{defi}
  
  It turns out that the safe and the risky $c$-capacities for $\Indep_b$ are the same, but at the moment we will only consider the safe capacity. 
  
  \begin{prop}\label{prop:indepupper}
  No rate $R>\frac{-b\log(1-c)+c\log(1-b)}{c}$ is safely $c$-achievable for $\Indep_b$.
  \end{prop}
  \begin{proof}
  Assume for contradiction that $R>\frac{-b\log(1-c)+c\log(1-b)}{c}$ is safely $c$-achievable for $\Indep_b$, and let $\pi$ be a safe $(n,Rn,\Indep_b(n),c,\epsilon)$-protocol. Let $\delta=R-\frac{-b\log(1-c)+c\log(1-b)}{c}$. We know from Theorem \ref{theo:generalupperbound} that 
  \[I(X;T)\leq \frac{-b\log(1-c)+c\log(1-b)}{c}n=(R-\delta)n.\]
  Now
  \[H(X|T)=H(X)-I(X;T)\geq Rn-(R-\delta)n=\delta n.\]
  By Fano's inequality (\ref{ineq:Fano}) we get that the probability of error for Frank's guess is
  \[P_e\geq \frac{\delta n-1}{nR}.\]
  Thus for sufficiently large $n_0$ and  sufficiently small $\epsilon$ we cannot have $P_e\leq \epsilon$. When $P_e>\epsilon$ there must exist an $x\in\mathcal{X}$ such that $\Pr(D(T)\neq x|X=x)>\epsilon$, so $R$ is not safely $c$-achievable. 
  \end{proof}

  Next we want to show that all rates $R<\frac{-b\log(1-c)+c\log(1-b)}{c}$ are safely $c$-achievable for $\Indep_b$. To do this, we can consider each person to be a usage of a channel, and use Shannon's Noisy-Channel Theorem.
  
  \begin{theo}\label{theo:indeplower}
  Any rate $R<\frac{-b\log(1-c)+c\log(1-b)}{c}$ is safely $c$-achievable for $\Indep_b$.
  \end{theo}
  \begin{proof}
  Let $R<\frac{-b\log(1-c)+c\log(1-b)}{c}$ and let $c'\leq c$ be a number such that $\frac{b(1-c')}{c'(1-b)}$ is rational and $R<\frac{-b\log(1-c')+c'\log(1-b)}{c'}$. Now use $b$ and $c'$ to define $a$ and $d$ as in Example \ref{exam:achieve}. We consider the channel that on input $j$ with probability $b$ returns a random uniformly distributed element in $\{1+(j-1)a,2+(j-1)a\dots,ja\} \mod{d}$, and with probability $1-b$ it returns a random and uniformly distributed element in $\{1,\dots ,d\}$. We see that each person sending a message, exactly corresponds to using this channel. The computation (\ref{eq:capa}) from Example \ref{exam:achieve} shows that when input of this channel is uniformly distributed, the mutual information between input and output is $\frac{-b\log(1-c')+c'\log(1-b)}{c'}$. Thus the capacity of the channel is at least this value (in fact, it is this value). We now use Shannon's Noisy-Channel Coding Theorem \cite{Shannon48, ThomasCover91} to get an error correcting code $\mathfrak{C}:\mathcal{X}\to \{1,\dots,d\}^n$ for this channel, that achieves rate $R$ and for each $x$ fails with probability $<\epsilon$. Now when $X=x$ any player that is not leaking will send a message chosen uniformly at random from $\{1,\dots,d\}$ and any player $\plr_i$ with $L_i=1$ chooses a message uniformly at random from $\{1+(j-1)a,2+(j-2)a,\dots, ja\}\mod{d}$, where $j=\mathfrak{C}(x)_i$ is the $i$'th letter in the codeword for $x$. This ensures that Frank will be able to guess $x$ with probability $1-\epsilon$.  We see that given $X$ the random variable $(A_i,L_i)$, is independent from $A_1,L_1,\dots, A_{i-1},L_{i-1},A_{i+1}, L_{i+1},\dots, A_n,L_n$. Using the computation from (\ref{eq:Bayes}) we now get that $\Pr(L_i=1|T=t,X=x)$ is either $0$ or $c'\leq c$ as needed. 
  \end{proof}

  For a specific code $\mathfrak{C}$, the message $A_i$ send by $\plr_i$ may not be uniform, as some letters might occur more often than others as the $i$'th letter in $\mathfrak{C}(X)$. On the other hand, given $L_i=0$, we know that $A_i$ is uniformly distributed, and Theorem \ref{theo:Alice} then implies that the expected increases in suspicion will be strictly greater than the leaked information. The computation (\ref{eq:Bayes}) shows that the expected increases in suspicion is the same no matter the distribution of $\mathfrak{C}_i$, but of course the amount of leaked information is greatest when $\mathfrak{C}_i$ is uniformly distributed.

  \begin{coro}\label{coro:indepcapa}
  The safe $c$-capacity for $\Indep_b$ is $\frac{-b\log(1-c)+c\log(1-b)}{c}$.
  \end{coro}
  \begin{proof}
  Follows from Proposition \ref{prop:indepupper} and Theorem \ref{theo:indeplower}.
  \end{proof}
  
  Corollary \ref{coro:indepcapa} shows that if you want information about something that some proportion $b$ of the population knows, but no one wants other people to think that they know it with probability $>c$, you can still get information about the subject, and at a rate of $\frac{-b\log(1-c)+c\log(1-b)}{c}$ bits per person you ask. What if only $l$ persons in the world have the information? They are allowed to blend into a group of any size $n$, and observers will think that any person in the larger group is as likely as anyone else to have the information. Only the number of persons with the information is known to everyone.
  
  If they are part of a group of $n\to \infty$ people, then each person in the larger group would have the information with probability $b=\frac{l}{n}$. If we forget that exactly $l$ persons know the information, and instead assumed that all the $L_i$s were independent with $\Pr(L_i=1)=b$ they would be able to leak
  \begin{align*}
  \frac{-b\log(1-c)+c\log(1-b)}{c}n=& \frac{-\frac{l}{n}\log(1-c)+c\log(1-\frac{l}{n})}{c}n\\
\to& \left(\frac{\log(1-c)}{c}-\log(e)\right)l
  \end{align*}
  bits of information, where $e$ is the base of the natural logarithm. We will see that even in the case where the number of leakers is known and constant, we can still get this rate. First we define the distribution of $(L_1,\dots, L_n)$ that we get in this case.
  
  \begin{defi}\label{defi:Fixed}
  Let $\Fixed(l,n)$ be the random variable $(L_1,\dots, L_n)$ that is distributed such that the set of leakers $\{\plr_i|L_i=1\}$ is uniformly distributed over all subsets of $\{\plr_1,\dots, \plr_n\}$ of size $l$.
  
   A rate $R$ is \emph{safely/riskily $c$-achievable for $\Fixed$} if for all $\epsilon>0$ and all $l_0$, there exists a safe/risky $(n,lR,\Fixed(l,n),c,\epsilon)$-protocol for some $l\geq l_0$ and some $n$.
  
  The \emph{safe/risky $c$-capacity for $\Fixed$} is the supremum of all safely/riskily $c$-achievable rates for $\Fixed$.
  \end{defi}
  Notice that in this definition, the rate is measured in bits per leaker rather than bits per person communicating. That is because in this setup we assume that the number of people with the information is the bounded resource, and that they can find an arbitrarily large group of person to hide in.

  Again, it turns out that the safe and the risky $c$-capacity for $\Fixed$ are actually the same, but for the proofs it will be convenient to have both definitions.

  \begin{prop}\label{prop:safeupper}
  No rate $R>\frac{-\log(1-c)}{c}-\log(e)$, where $e$ is the base of the natural logarithm is safely $c$-achievable for $\Fixed$.
  \end{prop}
  \begin{proof}
  This proof is very similar to  the proof of Proposition \ref{prop:indepupper}.
  
  Assume for contradiction that $R>\frac{-\log(1-c)}{c}-\log(e)$ is safely $c$-achievable. Consider a safe $(n,lR,\Fixed(l,n),c,\epsilon)$-protocol $\pi$. We know from Theorem \ref{theo:generalupperbound} that 
  \[I(X;T)\leq \frac{-\frac{l}{n}\log(1-c)+c\log\left(1-\frac{l}{n}\right)}{c}n\leq l\left(\frac{-\log(1-c)}{c}-\log(e)\right).\]
  Here the second inequality follows from $\ln(1+x)\leq x$ or equivalently $\log(1+x)\leq \frac{x}{\ln(2)}=-x\log(e)$. Let $\delta:=R-\frac{-\log(1-c)}{c}-\log(e)$. Now
  \[H(X|T)=H(X)-I(X;T)\geq l\left(R-\frac{-\log(1-c)}{c}-\log(e)\right)=l\delta.\]
  By Fano's inequality we get that the probability of error,  $P_e=\Pr(D(t)\neq x)$ averages over all possible values of $x$ is
  \[P_e\geq \frac{l\delta-1}{lR}.\]
  Thus if we chose $l_0$ sufficiently large and $\epsilon$ sufficiently small we cannot have $l\geq l_0$ and $P_e\leq \epsilon$, so that there must be some value $x$ where the probability of error $\Pr(D(T)\neq x|X=x)$ is greater than $\epsilon$.
  \end{proof}

  \begin{theo}\label{theo:riskylower}
  Any rate $R<\frac{-\log(1-c)}{c}-\log(e)$ is riskily $c$-achievable for $\Fixed$.
  \end{theo}
  One way, and in the author's opinion the most illuminating way, to prove this is similar to the proof of Theorem \ref{theo:indeplower}. Again we would consider each player to be a use of a channel. However, in this case the different usages of the channel would not be independent as we know exactly how many people who are leaking. Intuitively, this should not be a problem, it should only make the channel more reliable. However to show that this work, we would have to go through the proof of Shannon Noisy-Channel Coding Theorem, and show that it still works. Instead we will give a shorter but less natural proof.

  The idea is to use the same protocol as when we showed the lower bound in Theorem \ref{theo:indeplower}. However, for each particular rate $R$ and number of player $n$, there is a small probability that Frank fail to guess $X$. The probability that exactly $b n$ players are leaking, when all the $L_i$'s are independent tends to $0$ as $n$ tends to infinity, so we could be unlucky that Frank often fais in this case. Instead of using the protocol from Theorem \ref{theo:indeplower} on all the players, we divide the player onto two groups and use Theorem \ref{theo:indeplower} on each group.

  \begin{proof}
 Let $R<\frac{-\log(1-c)}{c}-\log(e)$, then we can find rational $b>0$ and rational $c'< c$ such that $R<\frac{-b\log(1-c')+c'\log(1-b)}{bc'}$, and let $n_0,\epsilon>0$ be given. By Theorem \ref{theo:indeplower} for any $\epsilon'>0$ and any $n'_0$ there exists a safe $(n,nR,\Indep_b(n),c',\epsilon')$-protocol where $n>n'_0$. Take such a protocol, where $\epsilon'>0$ is sufficiently small and $n'_0$ is sufficiently large. We can also assume that $bn$ is an integer. 
 
 Now we will use this to make a risky $(2n,2\lceil nR\rceil,\Fixed(2nb),c,\epsilon)$-protocol. For such a protocol, $X$ should be uniformly distributed on $\{1,\dots, 2^{2\lceil nR\rceil}\}$, but instead we can also think of $X$ as a tuple $(X_1,X_2)$ where the $X_i$ are independent and each $X_i$ is uniformly distributed on $\{1,\dots ,2^{\lceil nR\rceil}\}$. Now we split the $2n$ persons into two groups of $n$, and let the first group use the protocol from the proof of Theorem \ref{theo:indeplower} to leak $X_1$, and the second group use the same protocol to leak $X_2$. We let Franks guess of the value of $X_1$ be a function $D_1$ depending only of the transcript of the communication of the first group, and his guess of $X_2$ be a function $D_2$ depending only on the transcript of the second group. These functions are the same as $D$ in the proof of Theorem \ref{theo:indeplower}. The total number of leakers is $2nb$, but the number of leakers in each half varies. Let $S_{\Indep}$ denote random variable that gives the number of leakers among $n$ people, when each is leaking with probability $b$, independently of each other. So $S_{\Indep}$ is binomially distributed, $S_{\Indep}\sim \text{B}(n,b)$. Let $S_{\Fixed,1}$ denote the number of leakers in the first group as chosen above. Now we have.
 
 \begin{lemm}
 For each $k$,
 \begin{align}
 \frac{\Pr(S_{\Fixed,1}=k)}{\Pr(S_{\Indep}=k)}\leq 2.\nonumber
 \end{align}
 \end{lemm}
 \begin{proof}
 We have 
 \[\Pr(S_{\Fixed,1}=k)=\frac{\binom{2l}{k}\binom{2n-2l}{n-k}}{\binom{2n}{n}}.\]
 A simple computation shows 
 \[\frac{\Pr(S_{\Fixed,1}=k)\Pr(S_{\Indep}=k+1)}{\Pr(S_{\Indep}=k)\Pr(S_{\Fixed,1}=k+1)}=\frac{n-2l+k+1}{2l-k}\frac{l}{n-l},\]
 which is $>1$ for $k\geq l$ and $<1$ for $k<l$. Thus for fixed $n$ and $l$ the ratio $ \frac{\Pr(S_{\Fixed,1}=k)}{\Pr(S_{\Indep}=k)}$ is maximized by $k=l$.
 Using Sterlings formula, 
 \[1\leq \frac{n!}{\sqrt{2\pi n}\left(\frac{n}{e}\right)^n}\leq \frac{e}{\sqrt{2\pi}}\]
 we get
  \begin{align}
 \frac{\Pr(S_{\Fixed,1}=l)}{\Pr(S_{\Indep}=l)}=& \frac{\binom{2l}{l}\binom{2n-2l}{n-l}n^n}{\binom{2n}{n}\binom{n}{l}l^l(n-l)^{n-l}}\nonumber\\
 \leq & \sqrt{2} \left(\frac{e}{\sqrt{2\pi}}\right)^3\nonumber\\
 <&2.\nonumber
 \end{align}
 \end{proof}


 
 Given that $S_{\Indep}=k=S_{\Fixed,1}$, the distribution on $(L_1,\dots ,L_n)$ and transcript is the same in the protocol for $\Indep_b$ as it is for the first group in the above protocol. As Franks guessing function is the same in the two cases, the probability of error given $S_{\Indep}=k=S_{\Fixed,1}$ is the same in the two protocols. Let $E_{k}$ denote the probability of error in the protocol for $\Indep_b$ given $S_{\Indep}=k$, and let $E_{\Fixed,1}$ denote the probability that Franks guess of $X_1$ is wrong.
 \begin{align*}
 E_{\Fixed,1}=&\sum_{k=1}^n \Pr(S_{\Fixed,1}=k)E_k\\
 \leq &\sum_{k=1}^n2\Pr(S_{\Indep}=k)E_k\\
 \leq & 2\epsilon'.
 \end{align*}
 By the same argument, the probability that Frank guess $X_2$ wrong is at most $2\epsilon'$, so the probability that he guess $X=(X_1,X_2)$ is at most $4\epsilon'$. By choosing a sufficiently low $\epsilon'$ this is less than $\epsilon/2$
 
 To compute the posterior probability $\Pr(L_i=1|T=t)$ that $\plr_i$ was leaking, we have to take the entire transcript from both groups into account. Given $T$ and $X$, let $K$ denote the set of players who sent a message consistent with knowing $X$, and let $|K|$ denote the cardinality of $K$. Let $S$ be the set of the $2l$ leaking players, and let $s$ be a set of $2l$ players. Now
 \begin{align*} 
 \Pr(S=s|X=x,T=t)=\frac{\Pr(T=t|S=s,X=x)\Pr(S=s|X=x)}{P(T=t|X=x)}.
 \end{align*}
 This is $0$ if $s$ contains players who send a message not consistent with having the information, and is constant for all other $s$. Thus any two players who send a message consistent with having the information, are equally likely to have known $X$ given $T$ and $X$, so they will have $\Pr(L_i=1|T=t,X=x)=\frac{2l}{|K|}$. 
 So to ensure that $\Pr(L_i=1|T=t,X=x)\leq c$ with high probability (for each $x$ and random $t$) we only need to ensure that with high probability, $|K|\geq \frac{2l}{c}$. We see that $|K|=2l+\text{B}\left(2n-2l,\frac{b(1-c')}{c'(1-b)}\right)$, which have expectation $2l+(2n-2l)\frac{b(1-c')}{c'(1-b)}=\frac{2l}{c'}=\frac{2l}{c}+2l\frac{c-c'}{cc'}$. We also see that the variance is $(2n-2l)b(1-b)$, so for sufficiently high $n$ (and thus $l$) Chebyshev's inequality, shows that $|K|\geq \frac{2l}{c}$ with probability $1-\epsilon/2$. Thus for sufficiently large $n_0'$ and sufficiently low $\epsilon'$, the resulting protocol is a risky $(2n,2\lceil nR\rceil,\Fixed(2nb),c,\epsilon)$-protocol.
  \end{proof}

  \subsection{General $\mathfrak{L}$-structures}
  
  We have shown that the safe $c$-capacity for $\Fixed$ is $\leq \frac{-\log(1-c)}{c}-\log(e)\leq$ the risky $c$-capacity $\Fixed$. To finish the proof that they are both $\frac{-\log(1-c)}{c}-\log(e)$, we only need to show that the safe capacity is not smaller than the risky. Notice that the corresponding claim is not true if we are only interested in the mutual information between $X$ and transcript $T$. Here we could construct a collaborating cryptography protocol where with very high probability, $\Pr(L_i=1|T=t)<10^{-100}$, and yet $I(X;T)\geq 10^{100}$. To do this we need to take $X$ to have extremely high entropy, and with a very low probability a leaking player will send $X$ in a message, and otherwise just send some fixed message. The point of this section is to show that you cannot do something similar for reliable leakage. We will prove this in a setting that generalise $\Indep_b$ and $\Fixed$. Remember that the difference between $\Indep_b$ and $\Fixed$ capacity is not only in the distributions on $(L_1,\dots L_n)$, but also in what we are trying to minimize the use of. In $\Indep_b$ we want to have as few people communicating as possible, while in $\Fixed$ we only care about the number of people who are leaking. Our general definition have to capture this difference as well.
  
  \begin{defi}\label{defi:Lstruc}
  An \emph{$\mathfrak{L}$-structure} $(\mathfrak{L},C)$ is a set $\mathfrak{L}$ of joint distributions of $(L_1,\dots,L_n)$ (where $n$ do not need to be the same for each element), where each $L_i$ is distributed on $\{0,1\}$, together with a \emph{cost function} $C:\mathfrak{L}\to \mathbb{R}_{\geq 0}$.  
 
 $\Indep_b$ is the $\mathfrak{L}$-structure $(\mathfrak{L}_{\Indep_b},C_\#)$, where $\mathfrak{L}_{\Indep_b}$ is the set of distributions on $(L_1,\dots,L_n)$ (over $n\in \mathbb{N}$) where for all $i$, $\Pr(L_i=1)=b$ and the $L_i$ are independent, and $C_\#$ is the function that sends a distribution on $(L_1,\dots,L_n)$ to $n$. 
 
 $\Fixed$ is the $\mathfrak{L}$-structure $(\mathfrak{L}_{\Fixed},C_{\Fixed})$ of distributions on $(L_1,\dots, L_n)$ such that for some number $l$ the set $\{\plr_i|L_i=1\}$ is uniformly distributed over all subsets of $\{\plr_1,\dots,\plr_n\}$ of size $l$, and $C_{\Fixed}$ sends a distribution on $(L_1,\dots, L_n)$ to this number $l$.
  
 For an $\mathfrak{L}$-structure $(\mathfrak{L},C)$ a rate $R$ is safely/riskily $c$-achievable for $(\mathfrak{L},C)$ if for all $\epsilon>0$ and all $h_0\geq 0$ there exists a safe/riskily $(n,h,L,c,\epsilon)$-protocol with $h\geq h_0, h\geq C(L)R$ and $L\in \mathfrak{L}$.
 
 The safe/risky $c$-capacity for $(\mathfrak{L},C)$ is the supremum of all safely/riskily $c$-achievable rates for $(\mathfrak{L},C)$.
  \end{defi}
  
  We see that Definition \ref{defi:Lstruc} agrees with Definition \ref{defi:Indep} and Definition \ref{defi:Fixed}, and is much more general. 
  
  \begin{prop}\label{prop:capismono}
  Let $(\mathfrak{L},C)$ be an $\mathfrak{L}$-structure. The safe $c$-capacity for $(\mathfrak{L},C)$ and the risky $c$-capacity for $(\mathfrak{L},C)$ are non-decreasing functions of $c$.
  \end{prop}
  \begin{proof}
  Let $c'>c$. Any safe/risky $(n,h,L,c,\epsilon)$-protocol is a safe/risky $(n,h,L,c',\epsilon)$-protocol, so any safe/riskily $c$-achievable rate for $(\mathfrak{L},C)$ is a  safe/riskily $c'$-achievable rate for $(\mathfrak{L},C)$.
  \end{proof}
  
  \begin{prop}\label{prop:riskygeqsafe}
  Let $(\mathfrak{L},C)$ be an $\mathfrak{L}$-structure. The safe $c$-capacity for $(\mathfrak{L},C)$ is at most the risky $c$-capacity for $(\mathfrak{L},C)$.
  \end{prop}
  \begin{proof}
  Any safe $(n,h,L,c,\epsilon)$-protocol is a risky $(n,h,L,c,\epsilon)$-protocol, so any safely $c$-achievable rate for $(\mathfrak{L},C)$ is riskily $c$-achievable for $(\mathfrak{L},C)$. 
  \end{proof}
  
  The opposite inequality almost holds. Before we show that, we need a lemma.
 
  \begin{lemm}\label{lemm:nosurprices}
  For any risky $(n,h,L,c,\epsilon)$-protocol $\pi$, there is a risky $(n,h,L,c,\epsilon)$-protocol $\pi'$ where each message is either $0$ or $1$, and given previous transcript and given that the person sending the message is not leaking, there is at least probability $1/3$ of the message being $0$ and at least $1/3$ of it being $1$. 
  \end{lemm}
  \begin{proof}

  To restrict to $\{0,1\}$ we simply send one bit at a time, so now we only have to ensure that the probability of a message sent by a non-leaker being $0$ is always in $[\frac{1}{3},\frac{2}{3}]$. If the next message is $0$ with probability $p<1/3$, given that the sender is not leaking we modify the protocol (the case where $p>2/3$ is similar). First, the player $\plr_i$ sending the message decides if she would have send $0$ or $1$ in the old protocol $\pi$. Call this message $a$. If $a=0$ she chooses a number in the interval $(0,p)$ uniformly at random, if $a=1$ she chooses a number in $(p,1)$ uniformly at random. She then sends the bits of the number one bit at a time until
  \begin{itemize}
\item She says $1$, or
\item Given transcript until now, there is probability $\geq\frac{1}{3}$ that $a=0$
\end{itemize}
In the first case we then know that $a=1$, and we can go to the next round of $\pi$. Each time $\plr_i$ says $0$, she doubles the probability that $a=0$, so if we are in the second case (and was not before the last message), $\Pr(a=0|T)<\frac{2}{3}$. In this case she will simply reveal $a$ in the next message. 

 Instead of choosing a real number uniformly from $(0,p)$ or $(p,1)$, which would require access to randomness with infinite entropy, $\plr_i$ can just in each step compute the probabilities of sending $0$ or $1$ given that she had chosen such a number. Thus if for every probability $p'$ every player has access to a coin that ends head up with probability $p'$, they only need a finite number of coin flips to follow the above protocol.  
  \end{proof}
   
   The following lemma almost says that the safe $c$-capacity for $(\mathfrak{L},C)$ is the same as the risky $c$-capacity for $(\mathfrak{L},C)$.
   
  \begin{lemm}\label{lemm:skreweq}\label{pretend}
  Let $c'>c$. The safe $c'$-capacity for $(\mathfrak{L},C)$ is at least the same as the risky $c$-capacity for $(\mathfrak{L},C)$.
  \end{lemm}
  \begin{proof}
  To show this, it is enough to show that if $R$ is a riskily $c$-achievable rate for $(\mathfrak{L},C)$, then $R$ is safely $c'$-achievable for $(\mathfrak{L},C)$. Let $R$ be a riskily $c$-achievable rate for $(\mathfrak{L},C)$, and let $\epsilon'>0$ and $h_0'$ be given. We want to show that there exists a safe $(n',h',L,c',\epsilon')$-protocol with $h'\geq h_0'$, $L\in \mathfrak{L}$ and $h'\geq C(L)R$. 
  
  As $R$ is riskily $c$-achievable for $(\mathfrak{L},C)$, there exists a risky $(n,h,L,c,\epsilon)$-protocol for any $\epsilon>0$ and some $L\in\mathfrak{L}$, $h\geq h_0'$, $h\geq C(L)R$ and $n$. Let $\pi$ be such a protocol, where $\epsilon$ is a small number to be specified later. 
  
  We want to modify $\pi$ to make it a safe protocol $\pi'$. First, by Lemma \ref{lemm:nosurprices} we can assume that all messages send in $\pi$ are in $\{0,1\}$ and given that the sender is not leaking, it has probability at least $1/3$ of being $0$ and at least probability $1/3$ of being $1$. 
  
To ensure that for no transcript $t$ and player $\plr_i$ we have $\Pr(L_i=1|X=x,T=t)>c'$, we modify the protocol, such that everyone starts to pretends ignorance if the next message could result in $\Pr(L_i=1|X=x,T^{k+1}=t^{k+1})>c'$. Formally, we define a protocol $\pi'$ that starts of as $\pi$ but if at some point the transcript is $t^k$ and for some $i$ and $b\in\{0,1\}$ we have $\Pr(L_i=1|T^{k+1}=t^k\circ b,X=x)>c'$ all the players \emph{pretends ignorance}, that is for the rest of the protocol they send messages as if they did not have the information and were following $\pi$. Notice that only the players who knows the information $x$ can decide if they should pretend ignorance, but this is not a problem as the players who do not have the information, is already sending messages as if they did not have the information.
  
 First we want to show that $\pi'$ is $c'$-safe. As long as they do not pretend ignorance we know that $\Pr(L_i=1|T^k=t^k,X=x)\leq c'$ for the partial transcript $t^k$ and all $i$. If at some point they starts to pretend ignorance, we have $\Pr(L_i=1|T^k=t^k,X=x)\leq c'$ before they start, and all messages will be chosen as if no one had the information. Eve, who knows $X$, can compute $\Pr(L_i=1|T^{k+1}=t^k\circ b,X=x)>c'$ for each $i$ and $b$, so she knows if everyone is pretending ignorance. Thus, Eve does not learn anything about $L$ from listening to the rest of the communication, so we will still have $\Pr(L_i=1|T=t,X=x)\leq c'$ when $\pi'$ terminates.
  
  Fix $x\in \mathcal{X}$. We want to compute the probability that they pretends ignorance given $X=x$. Let $E_{par,>c'}$ denote the event that for transcript $T$ from the execution of $\pi$, we can find some $k$ and some $i$ such that we have $\Pr(L_i=1|T^k=t^k,X=x)>c'$. That is, at some point in the execution of $\pi$, an observer would say that $\plr_i$ was leaking with probability $>c'$.  Let $E_{tot,>c}$ be that event that for the total transcript there is some $i$ such that $\Pr(L_i=1|T=t,X=x)>c$. For each transcript $t$ where $\Pr(L_i=1|T^k=t^k,X=x)>c'$ for some $k,i$, we consider that smallest $k$ such that $\Pr(L_i=1|T^k=t^k,X=x)>c'$ happens for some $i$. For this fixed $t^k$ let $T^{-k}$ denote the random variable that is distributed as the rest of the transcript given that the transcript starts with $t^k$ and $X=x$. Let $S_{t^k}$ denote the random variable 
  \[S_{t^k}=\Pr(L_i=1|T=t^k\circ T^{-k},X=x).\]
  That is, $S_{t^k}$ is a function of $T^{-k}$. We see that $S_{t^k}$ takes values in $[0,1]$ and $\E S_{t^k} =\Pr(L_i=1|T^k=t^k,X=x)>c'$ so by Markov's inequality on $1-S_{t^k}$ we get 
   \[\Pr(1-S_{t^k}\geq 1-c-\epsilon_1|X=x)\leq \frac{\E (1-S_{t^k})}{1-c-\epsilon_1}< \frac{1-c'}{1-c-\epsilon_1}\]
  for all $\epsilon_1>0$. Thus, given that $E_{par,>c'}$ happens, $E_{tot,>c}$ will happen with probability $\geq \frac{c'-c}{1-c}>0$. So $\frac{c'-c}{1-c} \Pr( E_{par,>c'}|X=x)\leq \Pr( E_{tot,>c}|X=x)\leq  \epsilon$, where the last inequality follows from the assumption about $\pi$. 
  
  Let $E_{ig}$ be the event that in the evaluation of $\pi'$ the players pretends ignorance. The players only pretends ignorance if they are one message away from making $E_{par,>c'}$ happen. We assumed that in $\pi$ each possible message get sent with probability at least $1/3$ if the sender is not leaking. As there is probability at least $1-c'$ that he is not leaking, each possible message gets sent with probability at least $\geq \frac{1-c'}{3}$
  %
   so $ \frac{1-c'}{3}\Pr(E_{ig}|X=x)\leq \Pr(E_{par,>c'}|X=x)$. Thus 
   \[\Pr(E_{ig}|X=x)\leq  \frac{3}{1-c'}\Pr(E_{par,>c'}|X=x)\leq 3\epsilon\frac{(1-c)}{(c'-c)(1-c')}.\]
  
 Let $T'$ denote the random variable you get from running $\pi'$ and $T$ the random variable you get from running $\pi$, with a joint distribution in such a way that $(X,L,T)=(X,L,T')$ unless the players pretends ignorance.
 We need to show that there is a decoding function $D'$ from the set of complete transcripts to possible values of $X$ such that for each $x$, $\Pr(D'(T')=x|X=x)\geq 1-\epsilon'$. 
 From the assumptions about $\pi$ we know that there is a function $D$ from the set of possible transcripts to the support of $X$ such that for each $x$, $\Pr(D(T)=x|X=x)\geq 1-\epsilon$. 
 We know that in $\pi'$ and for fixed $x$, the players only pretends ignorance with probability $\leq \frac{3\epsilon(1-c)}{(c'-c)(1-c')}$, so by setting $D'=D$ we get $\Pr(D'(T')=x|X=x)\geq 1-\epsilon-\frac{3\epsilon(1-c)}{(c'-c)(1-c')}$. 
 For sufficiently small $\epsilon$ (depending only on $c$ and $c'$) this is less than $\epsilon'$ and we are done. 
  \end{proof}
If we add a continuity assumption, we get that the safe and the risky $c$ capacity are the same.
\begin{coro}\label{coro:sameifcont}
Let $(\mathfrak{L},C)$ be a $\mathcal{L}$-structure. If the safe $c$-capacity for $(\mathfrak{L},C)$ as a function of $c$ is right-continuous at $c_0$, or if the risky $c$-capacity for $(\mathfrak{L},C)$ as a function of $c$ is left-continuous at $c_0$ then the safe $c_0$-capacity for $(\mathfrak{L},C)$ and the risky $c_0$-capacity for $(\mathfrak{L},C)$ are the same.
\end{coro}
\begin{proof}
Assume that the safe $c$-capacity for $(\mathfrak{L},C)$ as a function of $c$ a right-continuous at $c_0$. Then Lemma \ref{lemm:skreweq} shows that the risky $c$-capacity for $(\mathfrak{L},C)$ is at most the safe $c'$-capacity for $(\mathfrak{L},C)$ for all $c'>c$. By continuity assumption, this gives us that the risky $c$-capacity for $(\mathfrak{L},C)$ is at most the safe $c$-capacity for $(\mathfrak{L},C)$. Proposition \ref{prop:riskygeqsafe} shows the opposite inequality. The proof of the second part of the corollary is similar.
\end{proof}

\begin{coro}\label{coro:sameexceptcount}
Let $(\mathfrak{L},C)$ be a $\mathcal{L}$-structure. The safe $c$-capacity for $(\mathfrak{L},C)$ and the risky $c$-capacity for $(\mathfrak{L},C)$ are the same for all but at most countably many values $c\in(0,1)$. 
\end{coro}
\begin{proof}
By Proposition \ref{prop:capismono}, the safe $c$-capacity for $(\mathfrak{L},C)$ is a monotone function, so it is continuous in all but countably many points. Now \ref{coro:sameifcont} implies that it is the same as the risky $c$-capacity for $(\mathfrak{L},C)$ in all but countably many points.
\end{proof}

As promised, we can now show that for $\Indep_b$ the safe and the risky $c$-capacities are the same.
\begin{coro}
The safe $c$-capacity for $\Indep_b$ and the risky $c$-capacity for $\Indep_b$ are the same for all $c\in (0,1)$. 
\end{coro}
\begin{proof}
We know from Corollary \ref{coro:indepcapa} that the safe $c$-capacity for $\Indep_b$ is a continous function of $c$. Now Corollary \ref{coro:sameifcont} implies that it is the same as the risky $c$-capacity for $\Indep_b$. 
\end{proof}

\begin{coro}\label{coro:fixedcapa}
Let $c\in (0,1)$. The safe $c$-capacity for $\Fixed$ and the risky $c$-capacity for $\Fixed$ are both $\frac{-\log(1-c)}{c}-\log(e)$.
\end{coro}
\begin{proof}
We know from Proposition \ref{prop:safeupper} that the safe $c$-capacity for $\Fixed$ is at most $\frac{-\log(1-c)}{c}-\log(e)$, we know from Theorem \ref{theo:riskylower} that the risky $c'$ fixed capacity is at least $\frac{-\log(1-c)}{c}-\log(e)$, and from Corollary \ref{coro:sameexceptcount} that they are the same except on at most countably many values. Thus they must both be $\frac{-\log(1-c)}{c}-\log(e)$ on all but countably many values. We know from \ref{prop:capismono} that both are monotone, so they must both be $\frac{-\log(1-c)}{c}-\log(e)$ without exceptions.
\end{proof}

\section{The original cryptogenography problem}\label{sec:org}

  In \cite{Brody14} the authors studied the following cryptogenographic problem. We flip a coin, and tell the result to one out of $n$ people. The $n-1$ other people do not know who got the information. Formally that means we take $L=(L_1,\dots,L_n)$ to be the random variable that is uniformly distributed over all $\{0,1\}$-vectors $(l_1,\dots, l_n)$ containing exactly one $1$ and take $X$ to be uniformly distributed over $\{0,1\}$ independently from $L$.  We let the group of $n$ people use any collaborating cryptogenography protocol, and afterwards we let Frank guess the result of the coin flip (his guess depends only on the transcript) and then let Eve guess who was leaking (her guess can depend on both transcript and Franks guess). Eve wins if she guess the leaker or if Frank does not guess the result of the coin flip. Otherwise Frank and the $n$ people communicating wins. We assume that both Frank and Eve make there guess to maximise the probability that they win, rather than maximise the probability of being correct.\footnote{For example if $\Pr(L_1=1,X=0|T=t)=0.97$, $\Pr(L_1=1,X=1|T=t)=0.01$ and $\Pr(L_2=1,X=1|T=t)=0.02$ then it is most likely that $X=0$. However Frank will guess that $X=1$, even though it is much more likely that $X=0$. If Frank instead guessed $X=0$ then Eve would guess that $\plr_1$ is leaking and then Eve would be certain to win. Once Frank have guesses $X=1$, Eve will guess that $\plr_2$ is leaking even though it is much more likely that $\plr_1$ is leaking. This is because, given that Frank is correct, it is more likely that $\plr_2$ is leaking, and Eve do not care if she guess correct when Frank is wrong.}
  
  In \cite{Brody14} it was shown that the probability that the group wins is below $3/4$ and for sufficiently high $n$ it is at least $0.5644$. In this section we will generalise the problem to a situation were more people are leaking and $X$ contains more information. It is obvious how to generalise $X$ to more information, we simply take $X$ to be uniformly distributed on $\{1,\dots, 2^{\lceil h\rceil}\}$. It is less obvious to generalise to more leakers. When more people are leaking, it would be unreasonable to require Eve to guess all the leakers. If this was the rule, one of the leaking players could just reveal himself as a leaker and say what $X$ is, while the rest of the leakers behave exactly as the non-leakers. Instead we let Eve guess at one person and if that person is leaking, she wins.  
  \begin{defi}
  For fixed values of $h$, number of leakers $l$ and number of communicating players $n>l$ and a collaborating cryptogenography protocol $\pi$, we let $\Succ(h,l,n,\pi)$ denote the probability that after the players communicate using protocol $\pi$, Frank will guess the correct value of $X$ but Eve's guess will not be a leaker, assuming the Frank and Eve each guess using the strategy that maximise their own chance of winning. We define
  \[\Succ(h,l,n)=\sup_{\pi}(\Succ(h,l,n,\pi)),\]
 where the supremum is over all collaborating cryptogenography protocols $\pi$. Finally we define
 \[\Succ(h,l)=\lim_{n\to\infty}\Succ(h,l,n).\]
  \end{defi}
In this section we will investigate the asymptotic behaviour of $\Succ(h,l)$ when at least one of $l$ and $h$ tends to infinity. First some propositions.

\begin{prop}
The probability that the communicating players wins the game does not change if Eve is told the value of $X$ before they starts to communicate.
\end{prop}
\begin{proof}
If Frank guesses the correct value of $X$, Eve was going to assume that that was the correct value anyway (as she wants to maximise the probability that she is correct given that Frank was correct), and if Frank guesses wrong, she would win anyway. 
\end{proof}
In the rest of this section, we will assume that Eve knows the value of $X$.

\begin{prop}\label{prop:decinh}
$\Succ(h,l,n)$ and $\Succ(h,l)$ are non-increasing in $h$. 
\end{prop}
\begin{proof}
Let $h>h'$ and let $\pi$ be a protocol for parameters $h,l,n$ and let the secret be denoted $X$. We construct a protocol $\pi'$ with parameters $h',l,n$ and secret denoted by $X'$. In the first round of $\pi'$, $\plr_1$ announce $h-h'$ independent and uniformly chosen bits $Y$, and from then on, everyone follows protocol $\pi$ for $X=X'\circ Y$. It is clear the $\Succ(h,l,n,\pi)\leq \Succ(h',l,n,\pi')$.
\end{proof}

\begin{prop}\label{prop:incinn}
$\Succ(h,l,n)$ is non-decreasing in $n$. 
\end{prop}
\begin{proof}
We use the elimination strategy used in \cite{Brody14}. Let $n'>n$ and let $\pi$ be a protocol for parameters $h,l,n$. We now construct a sequence of protocols $\pi'_k$ for parameters $h,l,n'$. In the protocol $\pi'_k$ each non-leaking player thinks of a uniformly chosen number in $\{1,\dots, k\}$. First everyone who thought of the number $1$ announce that and they are out, then everyone who thought of the number $2$ and so on, until only $n$ players a left. If two or more player thought of the same number, we migth end up with less then $n$ players left. In that case the leakers just announce themselves. If we are left with exactly $n$ players, we know that the $l$ leakers are still among them, and we have no further information about who they are. They then use protocol $\pi$, and win with probability $\Succ(h,l,n)$. As $k\to\infty$, the probability that two players thought of the same number tends to $0$, so $\Succ(h,l,n',\pi'_k)\to \Succ(h,l,n,\pi)$.
\end{proof}
    
  \begin{theo}\label{theo:lowerboundoriginal}
  For all $p\in (0,1)$,
  \[\liminf_{l\to\infty}\Succ\left(\left\lceil\left(\frac{-\log(p)}{1-p}-\log(e)\right)l\right\rceil,l\right)\geq p.\]
  \end{theo}
  \begin{proof}
  We know from Corollary \ref{coro:fixedcapa} that the safe $c$-capacity for $\Fixed$ is $\frac{-\log(1-c)}{c}-\log(e)$. If we let $\epsilon>0$, and use this Corollary for $c=1-p+\epsilon/2$ we get that for sufficiently high $l,n$ and $h=\left\lceil\left(\frac{-\log(p)}{1-p}-\log(e)\right)l\right\rceil$ there is a protocol $\pi$ that will make Frank's probability of guessing wrong at most $\epsilon/2$, and seen from Eve's prespective, no one is leaking with probability $>1-p+\epsilon/2$. By the union bound, the probability that Frank is wrong or Eve is correct\footnote{Here we assume that Frank guess on the most likely value of $X$, and we allow Eve to use any strategy. It could be that Frank could do better, but he is guarantied at least this probability of winning.} is at most $\epsilon/2+1-p+\epsilon/2$, thus the communicating players win with probability at least $p-\epsilon$. 
    \end{proof}
    In particular we have
    \begin{coro}
    Let $l\to\infty$ and $h=h(l)$ be a function of $l$ with $h=o(l)$. Then $\Succ(h,l)\to 1$.
    \end{coro}
    \begin{proof}
    Follows from Theorem \ref{theo:lowerboundoriginal} and Proposition \ref{prop:decinh} 
    \end{proof}

    \begin{defi}
    Let the distribution of $(X,L_1,\dots,L_n)$ be given and let $\pi$ be a protocol with transcript $T$ and $\pi'$ a protocol with transcript $\pi'$. For a transcript $t$ of $\pi$ let $\mu_t$ denote the distribution $(X,L_1,\dots ,L_n)|_{T=t}$, and similar for transcripts $t'$ of $\pi'$. We say that $\pi$ and $\pi'$ are \emph{equivalent for $(X,L_1,\dots, L_n)$} (or just \emph{equivalent} when it is clear what the distribution of $(X,L_1,\dots, L_n)$ is) if the distribution of $\mu_T$ is the same as the distribution of $\mu_{T'}$.
    \end{defi}
    
    That is, the probability that the posterior distribution of $(X,L_1,\dots,L_n)$ is $\mu$ has to be the same for both $\pi$ and $\pi'$. Notice that for two different distributions of $(X,L_1,\dots ,L_n)$ with the same support, $\pi$ and $\pi'$ are equivalent for one of them if and only if they are equivalent for the other distribution. Thus, when the support of $(X,L_1,\dots, L_n)$ is clear, we can simply say equivalent.

    \begin{prop}\label{prop:equivalent}
    If $\pi$ and $\pi'$ are equivalent collaborating cryptogenography protocols, then $\Succ(h,l,n,\pi)=\Succ(h,l,n,\pi')$. 
    \end{prop}

    The next lemma show that we can ensure that before any player crosses probability $c$ of having the bit, seen from Eve's perspective, that player lands on this probability. 

    \begin{lemm}\label{lemm:stopatc}
    Let $\pi$ be any collaborating cryptogenography protocol, let $(X,L_1,\dots,L_n)$ have any distribution and let $c\in (0,1)$. Then there exists an equivalent collaborating cryptogenography protocol $\pi'$ such that when we use it on $(X,L_1,\dots L_n)$ and let $T'$ denote its transcript, it satisfies:
    For all $x\in\mathcal{X}$, all $\plr_i$ and all non-empty partial transcripts $t'^k$, if 
    \[\Pr(L_i=1|T'^k=t'^k,X=x) > c.\]
    then there is a $k'<k$ such that 
    \[\Pr(L_i=1|T'^{k'}=t'^{k'},X=x)=c\]
    \end{lemm}

    \begin{proof}
    Let $\pi$, $(X,L_1,\dots, L_n)$ and $c$ be given, and assume that $(x,i)=(x_0, i_0)$ is a counterexample to the requirement from the lemma. We will then construct a protocol $\pi'$ such that $(x_0,i_0)$ is not a counterexample for $\pi'$, and any $(x,i)$ that satisfied the requirement for $\pi$ also satisfy it for $\pi'$. By induction, this is enough to prove the lemma. 
    
    We can assume that the messages in $\pi$ are send one bit at a time. We say a partial transcript $t^k$ is problematic if 
    \[\Pr(L_{i_0}=1|T^{k}=t^{k},X=x_0)<c\]
    but 
    \[\Pr(L_{i_0}=1|T^{k+1}=t^k\circ m,X=x_0)\geq c.\]
    for some bit value $m$. Without loss of generality, assume that $m=1$. Let $p=\Pr(T_{k+1}=1|T^k=t'^k)$. 
    
    We will use the $c$-notation from from Section \ref{sec:mutual}, so for example
    \[c_{t^k, x_0}=\Pr(L_i=1|T^{k}=t^{k},X=x_0).\] 
    Now    
    \begin{align*}
    c&>c_{t^{k}, x_0}=p c_{t^k\circ 1, x_0}+(1-p) c_{t^k\circ 0, x_0}
    \end{align*}
    so $c_{t^k\circ 0,x_0}<c$. Let $q\in (p,1)$ be the number such that 
    \[c=qc_{t^k\circ 1,x_0}+(1-q)c_{t^k\circ 0,x_0}.\]
    Now we modify $\pi$. First, the player $\plr_j$, who is going to send to $k+1$'th message in $\pi$, decides if she would have send $0$ or $1$ in $\pi$. If she would have send $1$ she sends the bits $11$. If she would have send $0$ she send $10$ with probability $\frac{p(1-q)}{q(1-p)}\in (0,1)$, and otherwise she sends $00$. In all cases she sends the bits one at a time. They then continue the protocol $\pi$ as if only the last of the two bits had been send. If we let $T'$ denote the transcript of the protocol with this modification, we get 
    \[c_{T'^{k+1}=t^k\circ 0, x_0}=c_{T^{k+1}=t^k\circ 0,x_0}<c \]
    and
    \begin{align*}
    c_{T'^{k+1}=t^k\circ 1,x_0}=&\frac{p c_{T^{k+1}=t^k\circ 1, x_0}+(1-p)\frac{p(1-q)}{q(1-p)}c_{T^{k+1}=t^k\circ 0,x_0}}{p+(1-p)\frac{p(1-q)}{q(1-p)}}\\
    =&qc_{T^{k+1}=t^k\circ 0,x_0} + (1-q) c_{T^{k+1}=t^k\circ 1,x_0}\\
    =&c.    
    \end{align*}
    So if $\plr_j$ sends $11$ or $10$ in the modified protocol, we land on probability $c$. Let $\pi'$ be the protocol we get from $\pi$ by doing this modification for each problematic partial transcript $t^k$ in $\pi$. It is clear that $\pi$ and $\pi'$ are equivalent, and that any $(x,i)$ that satisfied the requirement before also do afterwards. 
    \end{proof}

 \begin{lemm}\label{lemm:upperboundoriginal}
 For any $c\in (0,1)$ and any $h,l,n,\pi$, we have $\Succ(h,l,n,\pi)\leq 1-\frac{ch+l\log(1-c)+lc\log(e)-c}{h}$.
 \end{lemm}
  
  \begin{proof}
  As $\Succ(h,l,n)$ is non-decreasing in $n$, we can assume that $n>\frac{l}{c}$, so that $\Pr(L_i=1)<c$ at the beginning. By Lemma \ref{lemm:stopatc} and Proposition \ref{prop:equivalent} we can assume that $\pi$ satisfy the requirement for $\pi'$ in \ref{lemm:stopatc}. 

  Let $\pi'$ be the protocol that starts of as $\pi$, but where the players starts to pretend ignorance (as in the proof of Lemma \ref{pretend}) if $\Pr(L_i=1|T^k=t^k,X=x)=c$ for some $i$, current transcript $t^k$ and the true value $x$ of $X$. This ensures that $\Pr(L_i=1|T'=t,X=x)$ for all $i$ and $t$. Let $T'$ be the transcript of $\pi'$. From Theorem \ref{theo:generalupperbound} we get 
  \[I(X;T')\leq \left(-\frac{\log(1-c)}{c}-\log(e)\right)l\]
  We let Frank guess as he would if we used protocol $\pi$. By Fano's inequality, (\ref{ineq:Fano}), Frank's probability of being wrong when he only see the transcript of $\pi'$ is
\begin{align*}
P_e\geq & \frac{H(X|T')-1}{\log(|\mathcal{X}|)}\\
= &\frac{H(X)-I(X;T')-1}{\log(|\mathcal{X}|)}\\
\geq & \frac{h-l\left(\frac{-\log(1-c)}{c}-\log(e)\right)-1}{h}
\end{align*}
In the cases where Frank are wrong in $\pi'$ there are two possibilities: Either the players did not pretend ignorance, in which case Frank would also be wrong if they used protocol $\pi$, or they did pretend ignorance so $\Pr(L_i=1|T^k=t^k,X=x)=c$ for some $i$ and some smallest $k$. When this first happens  Eve can just ignore all further messages in $\pi$ and guess that $\plr_i$ is leaking. This way she is wins with probability at least $c$. Thus, all the situations in $\pi'$ where Frank guesses wrong, corresponds to situations in $\pi$ where Eve would win with probability at least $c$. So Eve's probability of winning when the players are using protocol $\pi$ is at least
\[cP_e\geq \frac{ch+l\log(1-c)+lc\log(e)-c}{h}\]  
  \end{proof}

  \begin{theo}\label{theo:upperboundoriginal}
  Let $r>0$ be a real number. Now
  \[\limsup_{l\to \infty}\Succ(\lceil r\log(e)l\rceil,l)\leq \frac{\log(r+1)}{r\log(e)}\]
  \end{theo}
  \begin{proof}
  Set $c=\frac{r}{r+1}$ and $h=\lceil r\log(e) l\rceil$ in Lemma \ref{lemm:upperboundoriginal}. Then Eve's probability of winning is at least
  \[\frac{r\lceil r\log(e) l\rceil-l(r+1)\log(r+1)+lr\log(e)-r}{\lceil r\log(e) l\rceil (r+1)}\]
  As $l$ tends to infinity, this tends to 
  \[\frac{r^2\log(e) -(r+1)\log(r+1)+r\log(e)}{ r\log(e) (r+1)}=1-\frac{\log(r+1)}{r\log(e)}\]
  as wanted.
  \end{proof}
  
  In particular we have
  
  \begin{coro}
  Let $h\to\infty$ and let $l=l(h)$ be a function of $h$ with $l(h)=o(h)$. Then $\Succ(h,l)\to0$.
  \end{coro}
      \begin{proof}
    Follows from Theorem \ref{theo:upperboundoriginal} and Proposition \ref{prop:decinh}. 
    \end{proof}

\section{Hiding among innocents}\label{app:construction}\label{sec:innocents}

Until now we have assumed, that even the players who are not trying to leak information will collaborate. In this section we will show that we do not need the non-leakers to collaborate. As long as some people are communicating innocently, and that communication is sufficiently non-deterministic, we can use these people as if they were collaborating. 

Formally, we model the innocent communication by an innocent communication protocol. While protocols usually are designed to compute some function, innocent communication protocols is a way of describing what is already going on. An \emph{innocent communication protocol} $\iota$ is a protocol that for each possible partial transcript $s^k$ and each player $i$ gives a finite set $\mathcal{A}_{i,s^k}$ of possible messages that that person can send in the next round, and a probability distribution on that set. In innocent communication protocols every person sends a message in each round. This assumption is not a restriction: if we have a protocol where only one players sends messages at a time, we can turn it into an innocent communication protocol, by requiring that all the other players sends the message ``no message'' with probability $1$. 
We will only be interested in innocent communication protocols that continues for infinitely many rounds. This assumption is of course unrealistic but in practice we only need it to be long. 

Let $S$ denote the random variable that is the infinite transcript we get from running $\iota$, and let $S^k$ denote the partial transcript of the first $k$ rounds. For a player $\plr_j$ and a partial transcript $s^k$ of the first $k$ rounds of $\iota$ we define
  \[p_{max,j}(s^k)=\max_a(\Pr(A_{j,s^k}=a)|S^k=s^k),\] 
  where $A_{j,s^k}$ is the message sent by $\plr_j$ in round $k+1$. We say that $\iota$ is \emph{informative} if for a random transcript $S$ and for each player $\prod_{k\in\mathbb{N}} p_{max,j}(S^k)=0$ with probability $1$. In other words, if at each round in the protocol you try to guess what message $\plr_j$ will send in the next round, then with probability $1$ you will eventually fail. Notice that the model for innocent communication here is equivalent to what is used in \cite{Hopper04}, and the definition of informative is almost the same as the definition of \emph{always informative} in \cite{Hopper04} when one player is communicating.\footnote{The difference is that in \cite{Hopper04}, $\prod_{k\in\mathbb{N}} p_{max,i}(T^k)$ have to go to $0$ exponentially fast.}

We say that a collaborating cryptogenography protocol $\pi$ is \emph{revealing} if there is a partial transcript $t^k$ and a player $\plr_j$ that is to send the next message $A$ when the transcript is $t^k$ and a message $a$ such that $\plr_j$ will send message $a$ with positive probability if $L_j=1$ but not if $L_j=0$. If this is not the case, we say that $\pi$ is \emph{non-revealing}.\footnote{A non-revealing protocol can also reveal who the leakers are. For example, if it is known that one person is leaking and all but one person sends a message that could not have been send by a leaker. However if $\Pr(L=(0,\dots,0))>0$ then a non-revealing protocol will never reveal anyone as a leaker.} The point in cryptogenography is to hide who is sending the information, so we are only interested in non-revealing protocols.

The main theorem of this section is

  \begin{theo}\label{theo:hai}
  Let $\pi$ be a non-revealing collaborating cryptogenography protocol, and let $\iota$ be an informative communication protocol. Then there exists a protocol $\proto{\pi}{\iota}$ that is equivalent to $\pi$, but where the non-leakers follow the protocol $\iota$.
  \end{theo}  
  
  \begin{proof}
We construct the protocol $\proto{\pi}{\iota}$ and a the same time an interpretation function $i$ that sends transcripts $s$ of $\proto{\pi}{\iota}$ to transcripts $t$ of $\pi$. We want them to satisfy.

\begin{enumerate}
\item For each partial transcript $s^k$ of $\proto{\pi}{\iota}$ and each player $\plr_j$, $\proto{\pi}{\iota}$ gives a probability distribution, depending only on $X,L_j,s^k$ and $j$ that $\plr_j$ will use to choose his next message. 
\item If $L_j=0$ then $\plr_j$ choose her messages in $\proto{\pi}{\iota}$ using the same distributions as in $\iota$.
\item The interpretation function $i$ sends (infinite) transcripts $s$ of $\proto{\pi}{\iota}$ to either transcripts $t$ of $\pi$ or to ``error''. The probability of error is $0$.
\item If $T$ denote the transcript of $\pi$ and $S$ denotes the transcript of $\proto{\pi}{\iota}$, then given that $i(S)$ is not error, $(X,L_1,\dots, L_n,i(S))$ is distributed as $(X,L_1,\dots, L_n,T)$.   \label{req:4} 
\item For each transcript $t$ of $\pi$, the random variable $(X,L_1,\dots, L_n)$ is independent from $S$ given $i(S)=t$. \label{req:5}
\end{enumerate}
Here the second requirement ensures that non-leakers can follow the protocol without knowing $X$ or $\pi$. In fact, unlike in the collaborating communication protocol, they might be thinking that everyone is just having an innocent conversation. Thus in $\proto{\pi}{\iota}$ we refer to the non-leakers as \emph{innocents}. Notice the important assumption that first the innocent communication protocol $\iota$ is defined and \emph{then} we create a protocol $\proto{\pi}{\iota}$ for leaking information on top of that. This corresponds to assuming that the non-leaking players either do not care about the leak, or that they are oblivious to the protocol. If $\iota$ was allowed to depend what the leakers does, the non-leaking players could try to prevent the leak, and it would be a very different problem.

The fourth of the above requirements tells us that $\proto{\pi}{\iota}$ reveals at least as much about $(X,L_1,\dots, L_n)$ as $\pi$ and the last requirement say that we do not learn anything more. This ensures that Frank and Eve, who both know $\proto{\pi}{\iota}$, learns exactly as much from the transcript of $\proto{\pi}{\iota}$ as they would from the transcript of $\pi$. 

\begin{prop}
If $\proto{\pi}{\iota}$ satisfy the above requirements, then $\proto{\pi}{\iota}$ and $\pi$ are equivalent. 
\end{prop}
\begin{proof}
$i$ gives error with probability $0$, so we can ignore all those cases. By requirement \ref{req:4}, $i(S)$ has the same distribution as $T$, and by requirement \ref{req:4} and \ref{req:5} the distribution $\mu_s$ of $(X,L_1,\dots,L_n)$ given $S=s$ equals the distribution $\mu_{i(s)}$.
\end{proof}

Before we construct the protocol $\proto{\pi}{\iota}$ we will define a function $i'$ that sends partial transcripts $s^{k'}$ of $\proto{\pi}{\iota}$ to tuples $(t^k,[y,z))$ where $t^k$ is a partial transcript of $\pi$, and $[y,z)\subset [0,1)$ is a half-open interval. When $i'(s^{k'})=(t^k,[y,z))$, we refer to $t^k$ as the interpretation of $s^{k'}$. Loosely speaking, the point of the interval is that not all message in $\iota$ are sufficiently unlikely that they can correspond to a message in $\pi$, so instead of interpreting them to a message in $\pi$, we store the information by remembering an interval. For an infinite transcript $s$, the function $i'$ will satisfy
\begin{enumerate}
\item $i'(\lambda)=(\lambda,[0,1))$, where $\lambda$ is the empty string
\item If $i'(s^{k'})=(t^k,[y,z))$ then either 
\begin{itemize}
\item $i'(s^{k'+1})=(t^{k}\circ m,[0,1))$ for some message $m$ in $\pi$, or
\item $i'(s^{k'+1})=(t^k,[y',z'))$, where $[y',z')\subseteq [y,z)$
\end{itemize}
\item If $i'(s^{k'})=(t^k,[y,z))$ and $t^k$ is a complete transcript for $\pi$, then $y=0$, $z=1$ and $i'(s^{k''})=(t^k,[0,1))$ for all $k''\geq k'$
\end{enumerate}
Thus every time we reveal one more round from the transcript $s$, we will either learn one message in $\pi$ from the interpretation of $s$, or the interval gets smaller or stays the same. If the interpretation of $s^{k'}$ is $t^k$, we let $j(s^{k'})$ denote the index of the player to send the next message in $\pi$ when the current transcript is $t^k$. When it is clear what $s^{k'}$ is, we write $j$ instead of $j(s^{k'})$. If $i(s^{k'})=(t^k,[y,z))$ and $i(s^{k'+1})=(t^k\circ m,[0,1))$ we say that at time $k'$ $\plr_{j(s^{k'})}$ finished sending the message $m$ in $\pi$ and at time $k'+1$ $\plr_{j(s^{k'+1})}$ start sending a new message in $\pi$. 

Let $\mathcal{A}_{t^k}$ denote the set of messages that $\plr_j$ could send in $\pi$ after transcript $t^k$, and choose some ordering on this set. We now define a function $f:[0,1)\to \mathcal{A}_{t^k}$ such that 
\[f^{-1}(a)=[\Pr(A<a|L_j=0),\Pr(A\leq a|L_j=0)).\]
By definition of innocent communication protocol, each message in $\iota$ is chosen from a finite set, but to explain the point of the function $f$, imagine for now that $\iota$ said that in the next round $\plr_j$ should send a random real uniformly from in $[0,1)$. We could now interpret that as the message  $f(x)\in \mathcal{A}_{t^k}$ in $\pi$. Then $\proto{\pi}{\iota}$ would say that if $\plr_j$ was innocent he should send a number uniformly from $[0,1)$ and if he was leaking, he should first choose $a\in\mathcal{A}_{t^k}$ using the distribution specified by $\pi$, and then send a number chosen uniformly at random from $f^{-1}(a)$. More generally, if $\iota$ said that $\plr_j$ should choose his next message $M$ from some continuous distribution on $\mathbb{R}$, we could take the quantile function given $L_j=0$ of the message
\[m\mapsto \Pr(M<m|L_j=0)\]
to turn it into a message that is uniform on $[0,1)$ given $L_j=0$. Unfortunately, there is only finitely many possible message for $\plr_j$ to sent in each round, so instead of getting a number out of the quantile function, we define a similar function to get an interval. 
Let $i'(s^{k'})=(t^k,[y,z))$ and let $\mathcal{M}_{j,s^{k'}}$ denote the set of possible message that $\plr_j$ can send in round $k'+1$ when transcript is $s^{k'}$ and choose some ordering on the set. Define $g:[y,z)\to \mathcal{M}_{j,s^{k'}}$ by
\[g^{-1}(m)=\{y+t(z-y)|t\in [\Pr(M<m|L_j=0), \Pr(M\leq m |L_j=0))\}.\] 
Thus instead of getting a number in $[0,1)$ out of $m\in \mathcal{M}_{j,s^{k'}}$, we get an interval $g^{-1}(m)$, whose length is proportional to the probability that an innocent player would send that message. If $g^{-1}(m)\subset f^{-1}(a)$ for some $a\in \mathcal{A}_{t^k}$ we say that $\plr_j$ send $a$ in $\pi$ and define $i'(s^{k'+1})=(t^k\circ a,[0,1))$. Otherwise, $\plr_j$ is not done sending his message and we define $i'(s^{k'+1})=(t^k,g^{-1}(m))$. Now if for some $k'$ we have $i'(s^{k'})=(t,[0,1))$ where $t$ is a complete transcript of $\pi$ we define $i'(s^{k''})=(t,[0,1))$ for all $k''>k'$ and $i(s)=t$. If for some $s$ no such $k'$ exists, we define $i(s)$ to give ``error''.

\begin{figure}[h]

\begin{tikzpicture}

\draw [ultra thick] (5,0) --(5,5);
\node [right] at (5,5.1) {$1$};
\draw [fill] (5,5) circle [radius=0.05];
\node [right] at (5,-0.1) {$0$};
\draw [fill] (5,0) circle [radius=0.05];
\node [left] at (5,0.8) {$\alpha$};
\draw [fill] (5,0.8) circle [radius=0.05];


\draw (1.5,3) circle [radius=1.5];
\node at (0.4,4.6) {$\mathcal{A}$};
\draw [fill] (1.5,3.5) circle [radius=0.025];
\draw [fill] (1.5,2.5) circle [radius=0.025]; 
\node [left] at (1.5,3.5) {$a^2$};
\node [left] at (1.5,2.5) {$a^1$};
\draw [->] (5,5) -- (3.5,3.5) --(1.7,3.5);
\draw (5,2) -- (3.5,3.5);
\draw [->] (5,2) -- (4,1.5) -- (1.7,2.402);
\draw (5,0) -- (4,1.5); 

\node [right] at (3,2) {$f$};
\node [above] at (3.2,3.5) {$f$};

\draw (8.5,4) circle [radius=1];
\draw (8.5,1) circle [radius=1];
\node at (9.5,5) {$\mathcal{M}_1$};
\node at (9.5,2) {$\mathcal{M}_2$};
\draw [fill] (8.5,4.4) circle [radius=0.025];
\node [right] at (8.5,4.4) {$m_1^2$}; 
\draw [fill] (8.5,3.6) circle [radius=0.025];
\node [right] at (8.5,3.6) {$m_1^1$}; 
\draw [fill] (8.5,1.4) circle [radius=0.025];
\node [right] at (8.5,1.4) {$m_2^2$}; 
\draw [fill] (8.5,0.6) circle [radius=0.025];
\node [right] at (8.5,0.6) {$m_2^1$}; 

\draw [->] (5,5) -- (6.5,4.2) -- (8.3,4.38);
\draw (5,3) -- (6.5,4.2);
\draw [->] (5,3) -- (6.5,2.6) -- (8.3,3.5);
\draw (5,0) -- (6.5,2.6);

\draw [->] (5,3) -- (5.7,2.2) -- (8.3,1.448);
\draw (5,1.8) -- (5.7,2.2);
\draw [->] (5,1.8) --(5.7,0.6) -- (8.3,0.6);
\draw (5,0) -- (5.7,0.6);

\node [above] at (7,4.3) {$g_1$};
\node at (7,3.1) {$g_1$};

\node [below] at (7,0.6) {$g_2$};
\node at (7,1.6) {$g_2$};

\end{tikzpicture}

\caption{Example of how to construct a part of $\proto{\pi}{\iota}$.\\
In this figure we see an example of how construct a part of $\proto{\pi}{\iota}$. In $\pi$, the next player to send a message is $\plr_j$. The message $A_1$ should come from $\mathcal{A}=\{a^1,a^2\}$. We have $\Pr(A_1=a^1|L_j=0)=0.4$, so $f:[0,1)\to \mathcal{A}$ sends $x\in [0,0.4)$ to $a^1$, and $x\in [0.4,1)$ to $a^2$. Now $L_j=1$, so $\plr_j$ first chooses a message from $\mathcal{A}$ to send, this happens to be $a^1$, and then a number $\alpha$ chosen randomly and uniformly from $f^{-1}(a^1)$. \\
In $\iota$, the next message $M_1$ that $\plr_j$ sends should be from $\mathcal{M}_1=\{m_1^1,m_1^2\}$. If $\plr_j$ was innocent and was following the protocol $\iota$, we would have $\Pr(M_1=m_1^1)=0.6$, so $g_1:[0,1)\to \mathcal{M}_1$ sends $x\in [0,0.6)$ to $m_1^1$ and the rest to $m_1^2$. As $\alpha\in [0, 0.6)$, $\plr_j$ now sends the message $m_1^1$. We see that $g_1^{-1}(m_1^1)$ overlaps with both $f^{-1}(a^1)$ and $f^{-1}(a^2)$, so and observer cannot yet determine which message in $\pi$ $\plr_i$ was sending, so $\plr_j$ has not send his message yet. His next message $M_2$ should be send from $\mathcal{M}_2=\{m_2^1,m_2^2\}$, and again it happens that if he was following $\iota$ then $\Pr(M_2=m_2^1)=0.6$, so $g_2:[0,0.6)\to \mathcal{M}_2$ sends $x\in [0,0.36)$ to $m_2^1$ and the rest to $m_2^2$. As $\alpha\in [0, 0.36)$, $\plr_j$ sends the message $m_2^1$, and now $g_2^{-1}(m_2^1)\subset f^{-1}(a^1)$, so now an observer can see that $\plr_j$ was sending the message $a^1$ in $\pi$, and $\plr_j$ is done sending his message in $\pi$.
}
\end{figure}

Next we define the protocol $\proto{\pi}{\iota}$. Any non-leaking player chooses his messages as given by $\iota$ and when the current transcript is $s^{k'}$ all players except $\plr_{j(s^{k'})}$ also choose their messages as in $\iota$. When a leaking player, $\plr_{j(s^{k'})}$, starts sending a message in $\pi$, he first chose the message $a\in \mathcal{A}_{t^k}$ using the distribution given by $\pi$ (this distribution depends on $X=x$). Next he chooses a number $\alpha$ randomly and uniform in $f^{-1}(a)$. Until he has send his message in $\pi$ he will now send messages $m$ such that $\alpha\in g^{-1}(m)$. This uniquely specifies which messages $m$ to send (notice that $g$ will depend on current transcript in $\proto{\pi}{\iota}$, so $m$ is not necessarily the same for every round). When we get to a transcript $s^{k'}$ that is interpreted as a complete transcript $t$ of $\pi$ all the players will just follow $\iota$.

We see that if in $\pi$ a leaking player's distribution of $a$ is exactly the same as a non-leaking players, then the distribution of the number $\alpha$ chosen by the leaking player in uniform on $[0,1)$. By the definition of $g$, the probability that a leaking player sends a particular message $m$ in $\proto{\pi}{\iota}$ is exactly the probability given by $\iota$, and thus the same as a non-leaking player. Using this reasoning in the opposite direction, this tells us that we can assume that even the innocents, when starting sending a message in $\pi$, chooses a uniformly distributed $\alpha\in [0,1)$ and sends the message $m$ such that $\alpha\in g(m)$, until they have send the message in $\pi$. They may not do that, but the probability of any transcript is the same as if they did. 
  
  Finally we need to check that $\proto{\pi}{\iota}$ satisfy the $5$ requirements. The first two follows from the construction. To show the third, we need to show that with for a random transcript $s$ of $\proto{\pi}{\iota}$ there will with probability $1$ exists a $k'$ such that $i'(s^{k'})=(t,[0,1))$ where $t$ is a complete transcript for $\pi$. As $\pi$ only have finitely many rounds, it is enough to show that for each message of $\pi$ we start sending in $\proto{\pi}{\iota}$, there is probability $1$ that we will finish sending it. Assume that $i'(s^{k'})=(t^k,[0,1))$ for some $k'$, where $t^k$ is an incomplete transcript of $\pi$, but for all $k''>k'$ the interpretation of $s^{k''}$ is still $t^k$. If $\plr_{j(s^{k'})}$ is innocent, everyone will be following $\iota$, so by the assumption that $\iota$ is informative, the set of transcripts where the length of the interval does not go to $0$ has probability $0$. As stated earlier we can assume that when sending a message in $\pi$, even the innocents starts by choosing a random number $\alpha$ uniformly from $[0,1)$. As $f$ only jumps in finitely many points, there is probability $0$ that $\plr_{j(s^{k'})}$ chooses one of these points. If he does not, and the length of the interval goes to $0$, he will eventually sent his message in $\pi$. Thus there is probability $0$ that a non-leaker does not send his message. A leaker chooses his random $\alpha\in [0,1)$ using a different distribution, but we can divide $[0,1)$ into a finite set of intervals (given by $f^{-1}(a)$) such that it is uniform on each of these intervals. This tells us that given $s^{k'}$ there is a constant $K$ such that, as long as $\plr_{j(s^{k'})}$ is still sending the same message in $\pi$, any continuation of the transcript is at most $K$ times more likely when $\plr_{j(s^{k'})}$ is leaking as when he is not leaking. Thus there is still probability $K\cdot 0=0$ that he will not finish his message in $\pi$.
  
 For the fourth requirement, we observe that any leaking player is actually choosing messages in $\pi$ following the distribution given by $\pi$, and then making sure that the message send in $\proto{\pi}{\iota}$ will be interpreted as the message he wanted to send in $\pi$. The innocent players are not doing this, but we have seen that the distribution on the message they send in $\proto{\pi}{\iota}$ are the same as if they did. Thus requirement $4$ holds. Finally we see that given $i(S)=t$ a player not sending a message in $\pi$ always follows $\iota$ and a player sending a message in $\pi$ can be thought of as haven chosen an $\alpha$ uniformly from $f^{-1}(a)$ where $a$ is the next message in transcript $t$. This is independent from $(X,L_1,\dots,L_n)$ and thus the last requirement follows. 
 \end{proof}
 
 To implement the protocol $\proto{\pi}{\iota}$ the leaking players do not have to chose all the infinitely many digits in a random number $\alpha$. Instead they can just for each message compute the probability that they would send each message, given that they had chosen an $\alpha$. We also see that if $i(S)$ does not give an error, then there is some $k$ such that $S^k$ determines $i(S)$. Thus for any particular $\proto{\pi}{\iota}$ and any $\epsilon>0$ there is a length $k$ such that, $i(S^k)$ gives a total transcript for $\pi$ with probability $>1-\epsilon$.  
 
 In order to find the protocol $\proto{\pi}{\iota}$ you need have a description of the protocol $\iota$. This is a strong assumption: even if you are able to communicate innocently, it does not mean that you are aware of the distribution you use to pick your random messages. In steganography, the weaker assumption that you have a random oracle that takes history and player index as input and gives a message following the innocent distribution as output, is sometimes enough \cite{Hopper04}. However, it is not clear if this weaker assumption is enough for the propose of cryptogenography. While it may not be possible to find $\iota$ for all kinds of innocent communications, there are situations where we can approximate $\iota$ very well. For example, if a person post blog posts, we can consider the message to be only parity of the minutes in the sending time. This value will probably, for most people, be close to uniformly distributed on $\{0,1\}$.

  \section{Open problems}\label{sec:open}
  In this paper we only considered how much information $l$ players can leak in an asymptotic sense, where $l$ tends to infinity, and the proof of the achievability results is not constructive. We have not tried to find any explicit protocols that work well for fixed specific values of $l$ and tolerance of errors $\epsilon$, but that would be an interesting possibility for further research. We assumed that both Eve and Frank knew the true distribution $q$ of $(X,L_1,\dots ,L_n)$. It might be interesting to consider the problem where their beliefs, $q_E$ and $q_F$ are different from $q$ and from each other. 
  
  We have only found the $c$-capacity for $\Fixed$ and for $\Indep_b$. It would be interesting to find a way to compute the capacity of more general $\mathfrak{L}$-structures.

  
  
  In the setup we considered here, there are two types of players. Some know the information that we want to leak and some do not. We could also imagine that some people know who knows the information, without knowing the information itself, and some could know who knows who knows the information and so on. We could also have people who would only know $X$ if it belongs to some set $S$, and otherwise only know that $X\notin S$. It is known from the game theory literature that all of this can be described by having a joint distribution $(X,P_1,\dots, P_n)$ where $X$ is the information we want to leak and $P_i$ is the random variable that player $i$ has as information \cite{Aumann99}.

A different generalisation would be to have players that tries to prevent the leakage by sending misleading information. Such players would also not want to be discovered. If Frank notice that someone is sending misleading information, he could just ignore all the messages send by that person.

  \section{Acknowledgements}
  
  I would like to thank my supervisors, Peter Keevash and S\o ren Riis for valuable discussions about cryptogenography. I also want to thank Jalaj Upadhyay for pointing me to \cite{Hopper04}.

    \bibliographystyle{plain}
\bibliography{susp}

	 \end{document}